\documentclass[
reprint,
superscriptaddress,
amsmath,amssymb,
aps,
prx,
floatfix
]{revtex4-1}

\usepackage[colorlinks=true,breaklinks=true,linkcolor=red,citecolor=magenta,urlcolor=magenta]{hyperref}

\usepackage{graphicx}
\usepackage{dcolumn}
\usepackage{bm}
\usepackage{dsfont}
\usepackage{amsthm}
\usepackage{chngcntr}
\usepackage{apptools}
\usepackage{amsmath}
\usepackage{mathtools}
\usepackage{graphicx}
\graphicspath{ {./images/} }
\usepackage{siunitx}
\usepackage[margin=1in,footskip=0.25in]{geometry}
\usepackage{blindtext}
\usepackage{enumitem}
\usepackage{xcolor}
\usepackage[makeroom]{cancel}

\usepackage{placeins}
\AtAppendix{\counterwithin{lemma}{section}}
\usepackage{longtable} 
\usepackage{algorithm}
\usepackage{algpseudocode}

\DeclareMathOperator*{\argmin}{arg\,min}

\usepackage[capitalise]{cleveref}
\crefname{figure}{Fig.}{Fig.}

\newtheorem{theorem}{Theorem}

\newtheorem{lemma}{Lemma}
\theoremstyle{definition}
\newtheorem{definition}{Definition}

\theoremstyle{definition}
\newtheorem{myalgorithm}{Algorithm}

\newtheorem{remark}{Remark}

\usepackage{tikz}
\usetikzlibrary{quantikz2}

\begin{document}

\title{Quantum Krylov Algorithm for Szeg\"o Quadrature}

\author{William~Kirby}
\email{william.kirby@ibm.com}
\affiliation{IBM Quantum, IBM T.J. Watson Research Center, Yorktown Heights, NY 10598, USA}

\author{Yizhi~Shen}
\affiliation{Applied Mathematics and Computational Research Division, Lawrence Berkeley National Laboratory, Berkeley, CA 94720, USA}

\author{Daan~Camps}
\affiliation{National Energy Research Scientific Computing Center, Lawrence Berkeley National Laboratory, Berkeley, CA 94720, USA}

\author{Anirban~Chowdhury}
\affiliation{IBM Quantum, IBM T.J. Watson Research Center, Yorktown Heights, NY 10598, USA}

\author{Katherine~Klymko}
\affiliation{National Energy Research Scientific Computing Center, Lawrence Berkeley National Laboratory, Berkeley, CA 94720, USA}

\author{Roel~Van~Beeumen}
\affiliation{Applied Mathematics and Computational Research Division, Lawrence Berkeley National Laboratory, Berkeley, CA 94720, USA}

\begin{abstract}

We present a quantum algorithm to evaluate matrix elements of functions of unitary operators.
The method is based on calculating quadrature nodes and weights using data collected from a quantum processor.
Given a unitary $U$ and quantum states $|\psi_0\rangle$, $|\psi_1\rangle$, the resulting quadrature rules form a functional that can then be used to classically approximate $\langle\psi_1|f(U)|\psi_0\rangle$ for any function $f$.
In particular, the algorithm calculates Szeg\"o quadrature rules, which, when $f$ is a Laurent polynomial, have the optimal relation between degree of $f$ and number of distinct quantum circuits required.
The unitary operator $U$ could approximate a time evolution, opening the door to applications like estimating properties of Hamiltonian spectra and Gibbs states, but more generally could be any operator implementable via a quantum circuit.
We expect this algorithm to be useful as a subroutine in other quantum algorithms, much like quantum signal processing or the quantum eigenvalue transformation of unitaries.
Key advantages of our algorithm are that it does not require approximating $f$ directly, via a series expansion or in any other way, and once the output functional has been constructed using the quantum algorithm, it can be applied to any $f$ classically after the fact.

\end{abstract}

\maketitle


\section{Introduction}

Within the field of quantum simulation, much effort has been devoted to specific computational physics problems, most notably simulation of time evolution and approximation of low-lying energies.
A cross-cutting algorithmic idea that has proven powerful within this literature is application of functions to operators, developed over a long sequence of papers with varying computational goals (see for example~\cite{poulin2009sampling,poulin2009preparing, hhl2009,childs2010quantumwalk,childs2012lcu,childs2017quantum,chowdhury2017quantum,vanApeldoorn2020quantumsdpsolvers,ge2019faster}).
Current state-of-the-art techniques in this area are quantum signal processing~\cite{low2017signalprocessing, low2019qubitization} and the quantum singular value transformation (QSVT)~\cite{gilyen2019qsvt}, which provide methods for applying functions to Hermitian operators, and more recently the quantum eigenvalue transformation of unitaries (QET-U)~\cite{dong2022groundstate}, which provides a method for applying functions to unitary operators.
It is no exaggeration to state that these techniques have been transformative in how the field thinks about algorithmic possibilities within quantum simulation.

In this work, we expand this lineage of methods by presenting a quantum algorithm to approximate matrix elements of the form $\langle\psi_1|f(U)|\psi_0\rangle$, where $U$ is a unitary operator that we assume can be efficiently implemented by a quantum circuit, $f$ is a function, and $|\psi_0\rangle,|\psi_1\rangle$ are quantum states.
The method uses the quantum computer to estimate matrix elements $U^j$ for a sequence of powers $j$ up to some maximum, then inserts the resulting data into matrices representing the projection of $U$ onto a Krylov space.
Up to this point, the algorithm is the same as quantum Krylov methods for estimating low-lying energies~\cite{mcclean2017subspace,colless2018computation,parrish2019filterdiagonalization,motta2020qite_qlanczos,takeshita2020subspace,huggins2020nonorthogonal,stair2020krylov,urbanek2020chemistry,kyriienko2020quantum,cohn2021filterdiagonalization,yoshioka2021virtualsubspace,epperly2021subspacediagonalization,seki2021powermethod,bespalova2021hamiltonian,baker2021lanczos,baker2021block,cortes2022krylov,klymko2022realtime,jamet2022greens,baek2023nonorthogonal,tkachenko2022davidson,lee2023sampling,zhang2023measurementefficient,kirby2023exactefficient,shen2023realtimekrylov,yang2023dualgse,yang2023shadow,ohkura2023leveraging,motta2023subspace,anderson2024solving,kirby2024analysis,yoshioka2024diagonalization,byrne2024super,yu2025quantum}.

The main idea of our algorithm is to use the entire spectrum of the projected $U$ to construct a \emph{quadrature rule}, i.e., an approximation of the form
\begin{equation}
\label{qsq_output_general}
    \langle\psi_1|f(U)|\psi_0\rangle\approx\sum_{k=0}^{d-1}\omega_kf(\lambda_k),
\end{equation}
where the $\lambda_k$ and $\omega_k$ are a collection of ``nodes'' and ``weights'' (respectively) defining the rule.
We can think of both sides of \eqref{qsq_output_general} as functionals of $f$, the output of the algorithm being the functional on the right-hand side, which is used to approximate the desired functional on the left-hand side.

Quadrature rules have a long history as approximations of integrals in classical numerical linear algebra~\cite{gauss1815methodus,szego1939orthogonal,wilf1962mathematics,jones1989moment}.
In our case, the complex unit circle is the appropriate domain for the integration since all eigenvalues of a unitary operator lie on it.
Optimal quadrature rules for integration on the complex unit circle are called \emph{Szeg\"o quadrature rules}, and these are the type of rule that our algorithm is able to output.
For this reason, we call our algorithm \emph{quantum Szeg\"o quadrature} (QSQ).

It is instructive to note the parallel between QSQ, QSVT~\cite{gilyen2019qsvt}, and QET-U~\cite{dong2022groundstate}, and then spell out the differences.
Like these prior techniques, QSQ provides a method for taking an input operator~--- in our case $U$~--- and transforming it under some matrix function $f$.
Since it transforms a unitary, it is most closely related to QET-U.
There are three main advantages of QSQ compared to prior work:
\begin{enumerate}
    \item Simplified circuits: while QET-U uses interleaved controlled-$U$s and single-qubit rotations, QSQ requires only Hadamard tests, like quantum Krylov algorithms for low-lying energies.
    In the general case, there is little practical difference between these two types of circuit, since a single controlled long-time evolution has similar cost to many controlled short-time evolutions combining to the same time.
    However, if certain symmetries are present then the Hadamard test can be simplified to replace the controlled unitary with an uncontrolled unitary~\cite{yoshioka2024diagonalization}.
    See \cref{fig:circ_comp} for a comparison of the circuit constructions (with credit to~\cite{dong2022groundstate} for a very similar figure).
    The comparison with QSVT is similar.
    \item QSQ does not require calculating a series approximation of $f$ on the way to the final approximation.
    Although the accuracy of our approximation can be related to certain series expansions, as we show in \cref{algorithm}, it is unnecessary to explicitly obtain those expansions.
    \item As noted above, once we have collected the data from the quantum part of the QSQ algorithm, we can approximate \eqref{qsq_output_general} up to the fixed degree $d$ for \emph{any} $f$ by only changing the classical post-processing.
    This is a consequence of the fact that the output of QSQ is a functional in the form of a quadrature rule.
    The accuracy of the approximation may vary for different $f$, but taking new quantum data for a new $f$ is not necessary.
\end{enumerate}

\begin{figure*}[t!]
    \centering
    \begin{quantikz}
        \lstick{$|0\rangle$} & \gate{H} & \ctrl{1} & \meter{X\text{ or }Y} \\
        \lstick{$|\psi_0\rangle$} & \qw & \gate{U^j} & \qw
    \end{quantikz}~\\
    (a)~\\
    \begin{quantikz}
        \lstick{$|0\rangle$} & \gate{e^{i\phi_0X}} & \ctrl{1} & \gate{e^{i\phi_1X}} & \ctrl{1} & \qw & \hspace{-0.15in}\cdots\hspace{0.04in} & \ctrl{1} & \gate{e^{i\phi_0X}} & \meter{X\text{ or }Y} \\
        \lstick{$|\psi_0\rangle$} & \qw & \gate{U} & \qw & \gate{U^\dagger} & \qw & \hspace{-0.15in}\cdots\hspace{0.04in} & \gate{U} & \qw & \qw
    \end{quantikz}~\\
    (b)
    \caption{Comparison of (a) Hadamard test circuit for QSQ and (b) circuit for QET-U~\cite{dong2022groundstate}.}
    \label{fig:circ_comp}
\end{figure*}
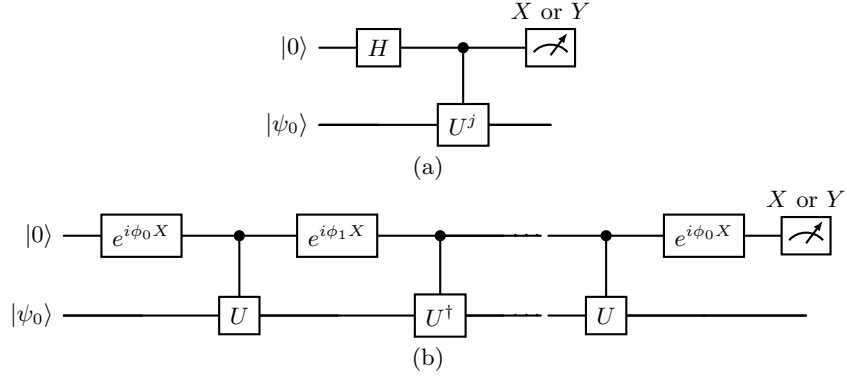

On the other hand, QET-U coherently constructs $f(U)|\psi_0\rangle$ in a quantum register (similarly $f(H)$ for QSVT), whereas QSQ approximates the functional as in \eqref{qsq_output_general} but does not explicitly construct the state $f(U)|\psi_0\rangle$.
In practice, many quantities of interest, such as expectation values and correlation functions, can be expressed in the form \eqref{qsq_output_general}.
But explicitly preparing the state $f(U)|\psi_0\rangle$ as in QET-U does provide the maximum flexibility in what to do with it.

This paper is organized as follows.
We begin by reviewing background on quadrature rules in \cref{quad_rules_background}.
In \cref{algorithm}, we present our main theoretical results: construction of the QSQ algorithm and proof that it yields a Szeg\"o quadrature rule.
In \cref{stable_alg}, we discuss the impact of noise and present a numerically-stable and noise-resilient implementation of our algorithm.
We present numerical demonstrations of our algorithm in \cref{numerics}.
Finally, we conclude with some context and discussion in \cref{discussion}.

\section{Background: quadrature rules}
\label{quad_rules_background}

As noted in the introduction, a quadrature rule $R$ is defined by a set of nodes $\lambda_k$ and weights $\omega_k$, and is used to approximate an integral $I$ by evaluating the integrand at the nodes and returning the weighted sum of the results:
\begin{equation}
    I(f)=\int_\mathcal{R}f(z)d\mu(z)\approx R(f)=\sum_{k=0}^{d-1}\omega_kf(\lambda_k).
    \label{quad_rule_gen}
\end{equation}
The nodes $\lambda_k$ and weights $\omega_k$ are determined by the range $\mathcal{R}$ and measure $\mu$ of the integral.
For example, a Riemann sum may be considered as a simple quadrature rule: an integral on the real line is approximated by partitioning the line into equal-sized intervals, evaluating the integrand at some points in the intervals, and weighting those values by the lengths of the intervals.
However, when one does not intend to take the number of nodes to infinity, a Riemann sum is typically not the optimal choice for approximation.
A better choice on the real line is \emph{Gauss quadrature}~\cite{gauss1815methodus}, which yields a prescribed set of nodes and weights for any measure (that is not too pathological) on an interval in the real line.
A degree-$d$ Gauss quadrature rule (i.e., having $d$ nodes and weights) is defined by being exact for polynomials up to degree $2d-1$, which is optimal~\cite{gauss1815methodus,wilf1962mathematics}.

\emph{Szeg\"o quadrature}~\cite{szego1939orthogonal,jones1989moment} (sometimes called \emph{Gauss-Szeg\"o quadrature}) is the cousin of Gauss quadrature for integration on the complex unit circle.
Instead of being exact for degree-$(2d-1)$ polynomials, it is exact for degree-$(d-1)$ \emph{Laurent polynomials}, i.e. functions of the form
\begin{equation}
\label{laurent_poly}
    f(z)=\sum_{j=-d+1}^{d-1}\alpha_jz^j.
\end{equation}
Note that, like a degree-$(2d-1)$ polynomial, this function contains $2d$ terms that are integer powers of the argument.
Hence it should perhaps be unsurprising that a Szeg\"o rule is optimal on the complex unit circle in the same sense as a Gauss rule on the real line, just for Laurent polynomials~\cite{jones1989moment}.

\section{Quantum Szeg\"o Quadrature Algorithm}
\label{algorithm}

\subsection{Construction}

In this section, we present the formal version of quantum Szeg\"o quadrature (QSQ) and give a proof of its optimality.
For this purpose, rather than \eqref{qsq_output_general}, we focus on approximating the simpler functional
\begin{equation}
\label{qsq_output}
    \langle\psi_0|f(U)|\psi_0\rangle.
\end{equation}
Linear combinations of the above with different states $|\psi_0\rangle$ can be used to construct the general functional \eqref{qsq_output_general}, as follows: let $a,b,c,d$ correspond to evaluating \eqref{qsq_output} with $|\psi_0\rangle$ replaced by $\frac{|\psi_0\rangle+|\psi_1\rangle}{\sqrt{2}}$, $\frac{|\psi_0\rangle-|\psi_1\rangle}{\sqrt{2}}$, $\frac{|\psi_0\rangle+i|\psi_1\rangle}{\sqrt{2}}$, and $\frac{|\psi_0\rangle-i|\psi_1\rangle}{\sqrt{2}}$, respectively.
Some algebra shows that
\begin{equation}
\label{gen_decomp}
    \frac{a-b + i(c-d)}{2} = \langle\psi_1|f(U)|\psi_0\rangle,
\end{equation}
the desired form \eqref{qsq_output_general}.

Given this construction, we can focus on the simpler form \eqref{qsq_output}.
Hence for the general form \eqref{gen_decomp} we would separately construct quadrature rules for $a$, $b$, $c$, and $d$, then form the linear combination of them given in \eqref{gen_decomp}.
Since each of the quadrature rules is itself linear in $f$, we can think of the linear combination \eqref{gen_decomp} of them as one big quadrature rule, justifying the form \eqref{qsq_output_general}.

We begin by defining a key subroutine: the \emph{quantum isometric Arnoldi's method}, a slight modification of the unitary quantum Krylov algorithm (also known as unitary variational quantum phase estimation~\cite{klymko2022realtime}).
Note that the following is an idealized version of the algorithm that is numerically unstable; the stable version is discussed in \cref{stable_alg}.

\begin{widetext}
\begin{myalgorithm}[``Quantum isometric Arnoldi's method"]~
\label{quantum_iso_arnoldi}
\begin{algorithmic}[1]

    \State\textbf{input} $n$-qubit unitary operator $U$, Krylov dimension $d$, initial state $|\psi_0\rangle$

    \State $\textbf{X}_{0}\gets1$

    \For{$j=1,2,\ldots,d$} \Comment{on quantum (Krylov circuits)}

    \State Calculate $\textbf{X}_{j}=\langle\psi_0|U^{j}|\psi_0\rangle\qquad$

    \State $\textbf{X}_{-j}\gets\overline{\textbf{X}_{j}}$

    \EndFor

    \For{$i,j=0,1,2,\ldots,d-1$} \Comment{build Krylov matrices $(\textbf{U},\textbf{S})$}

    \State $\textbf{U}_{ij}\gets\textbf{X}_{j-i+1}$

    \State $\textbf{S}_{ij}\gets\textbf{X}_{j-i}$

    \EndFor

    \State $\widetilde{\textbf{U}} \gets$ Gram-Schmidt orthonormalization of $\textbf{U}$ for Gram matrix $\textbf{S}$

    \State $\widetilde{\textbf{U}}_{,d-1} \gets \widetilde{\textbf{U}}_{,d-1}/\|\widetilde{\textbf{U}}_{,d-1}\|\qquad$ \Comment{$\widetilde{\textbf{U}}_{,d-1}=$ last column of $\widetilde{\textbf{U}}$}

    \State\textbf{return} $\widetilde{\textbf{U}}$
    
\end{algorithmic}
\end{myalgorithm}
\end{widetext}

This algorithm mimics the isometric variant of the classical Arnoldi's method~\cite{gragg1993positive}.
Prior to normalizing its last column in line 12 of \cref{quantum_iso_arnoldi} above, $\widetilde{\textbf{U}}$ is equal to the matrix that would be obtained by implementing Arnoldi's method starting from $|\psi_0\rangle$ with the iterator $U$.
This follows because by construction $\textbf{U}_{ij}=\langle\psi_i|U|\psi_j\rangle$ for $\psi_i\coloneqq U^i|\psi_0\rangle$, and $\textbf{S}_{ij}=\langle\psi_i|\psi_j\rangle$ is the Gram matrix of those basis vectors.
Gram-Schmidt orthonormalizing this basis using the inner products in $\textbf{S}$ entails orthogonalizing each $|\psi_i\rangle$ with respect to the $|\psi_j\rangle$ for $j=0,\ldots,i-1$, then normalizing, which is exactly what would have occurred if the $|\psi_i\rangle$ were generated using the Arnoldi iteration.

It follows that $\widetilde{\textbf{U}}$ is upper-Hessenberg, just like the matrix resulting from the classical Arnoldi method.
Normalizing the last column yields the so-called isometric Arnoldi method, where ``isometric'' refers to the fact that $\widetilde{\textbf{U}}$ is now also unitary, as shown in~\cite{gragg1993positive}.
In the limit of large $d$, the spectrum of $\widetilde{\textbf{U}}$ converges to that of $U$ unless breakdown occurs in the iteration: by fixing $\widetilde{\textbf{U}}$ itself to be unitary, we guarantee that the intermediate approximate spectra also lie on the complex unit circle.

As a quantum algorithm, the construction described above is preferable to orthonormalizing coherently on the quantum computer.
In a completely faithful reproduction of the Arnoldi iteration, one would ideally orthonormalize coherently, but that leads to much more complex quantum circuits~\cite{zhang2021quantum}.
However, implementing the orthonormalization in classical post-processing as we do introduces the potential for instability due to ill-conditioning of the Gram matrix $\textbf{S}$.
For the moment we ignore this and treat $\textbf{S}$ as if it is well-conditioned.

Unlike quantum Krylov algorithms for low energy approximation, we will not simply use the extremal (on the complex unit circle) eigenvalues of $\widetilde{\textbf{U}}$ to approximate the extremal eigenvalues of $U$.
Instead, we use the full spectrum of $\widetilde{\textbf{U}}^T$, as well as its eigenvectors, to approximate quantities of the form \eqref{qsq_output}.

We will begin by casting our simplified functional \eqref{qsq_output} as an integral.
Let $[N]\coloneqq\{0,1,2,\ldots,N-1\}$, and let
\begin{equation}
    |\psi_0\rangle=\sum_{j\in[N]}\gamma_j|\zeta_j\rangle ,
\end{equation}
be the decomposition of $|\psi_0\rangle$ in the eigenbasis of $\{|\zeta_j\rangle\}$ of $U$, for eigenvalues $\zeta_j$ in weakly-increasing order of phase in $\arg(\zeta_j)\in[-\pi,\pi)$.
If we wish to express $U$ as generated by a Hamiltonian $H$ over a time-interval $\Delta t$, then we have the correspondence
\begin{equation}
    \arg(\zeta_j)=-E_j\Delta t\mod[-\pi,\pi) ,
\end{equation}
between eigenenergies $E_j$ of $H$ and eigenvalues $\zeta_j$ of $U$, where $x\mod[-\pi,\pi)$ means $x+2\pi k$ for the integer $k$ such that $x+2\pi k\in[-\pi,\pi)$.
However, for the sake of generality we will primarily focus on $U$ and its eigendecomposition.
Note that we impose ordering on the $\zeta_j$ as described above, which does not necessarily correspond to ordering of the $E_j$; that will also be guaranteed if and only if $\Delta t$ is small enough that $E_j\Delta t\in[-\pi,\pi)$ for all $j$, which amounts to the constraint
\begin{equation}
    \Delta t\le\frac{\pi}{\|H\|}.
\end{equation}

Let $\mathds{T}$ denote the contour counterclockwise around the complex unit circle, and for each $j=0,1,2,\ldots,N-1$ let $\mathds{T}_j\subset\mathds{T}$ denote the arc within $\mathds{T}$ that lies between $\zeta_{j-1}$ and $\zeta_j$ (for the cases $j=0$ and $N-1$, we define $\zeta_{-1}=\zeta_{N-1}=-1$).
Using these definitions, we can express our functional \eqref{qsq_output} as an integral:
\begin{equation}
\label{riemann_stieltjes_form}
\begin{split}
    \langle\psi_0|f(U)|\psi_0\rangle&=\sum_{j\in[N]}|\gamma_j|^2f(\zeta_j)\\
    &=I(f)\coloneqq\int_\mathds{T}f(z)d\mu(z),
\end{split}
\end{equation}
where $\mathds{T}$ denotes the contour counterclockwise around the complex unit circle, and $I$ is a Riemann-Stieltjes integral whose measure is
\begin{equation}
\label{linear_mu}
    \mu(z)=\sum_{k\in[j]}|\gamma_k|^2\qquad\text{if}\quad z\in\mathds{T}_j.
\end{equation}
In words, if we begin at $z=-1\in\mathds{T}$, $\mu(z)$ is stepwise increasing around $\mathds{T}$, starting from $0$ and increasing to $\mu(z)=\sum_{k=0}^{N-1}|\gamma_k|^2=1$ when $z\in\mathds{T}_N$.
Hence
\begin{equation}
    d\mu(z)=\sum_{j\in[N]}|\gamma_j|^2\delta(z-\zeta_j),
\end{equation}
which explains why \eqref{riemann_stieltjes_form} holds.

The main idea of QSQ is to approximate $I$ using Szeg\"o quadrature with nodes $\lambda_k$ and weights $\omega_k$ computed using the quantum isometric Arnoldi method.
As a reminder, a quadrature rule $R(f)$ for an integral $I(f)$ is a set of pairs $(\lambda_k,\omega_k)$ such that
\begin{equation}
\label{szego_quadrature}
    I(f)\approx R(f)\coloneqq\sum_{k\in[d]}\omega_kf(\lambda_k).
\end{equation}
Szeg\"o quadrature refers to the case when the approximation \eqref{szego_quadrature} is exact for any Laurent polynomial $f$ of degree up to $d-1$, i.e., $f$ satisfies \eqref{laurent_poly} for some $\alpha_j$.
This is the optimal degree for which the expression \eqref{szego_quadrature} can be exact for a given number $d$ of node-weight pairs.
We prove the following theorem in \cref{quantum_iso_arnoldi_appendix}, employing results from~\cite{jones1989moment,gragg1993positive}:

\begin{theorem}
\label{main_thm}
    The nodes $\lambda_k$ of the Szeg\"o quadrature rule \eqref{szego_quadrature} are the eigenvalues of $\widetilde{\textbf{U}}^T$.
    The corresponding weights $\omega_k$ are given by
    \begin{equation}
    \label{weights_formula}
        \omega_k=|v_{k,0}|^2,
    \end{equation}
    where $v_{k,0}$ is the zeroth entry of the normalized eigenvector $v_k$ of $\widetilde{\textbf{U}}^T$ with eigenvalue $\lambda_k$.
\end{theorem}

Combining \eqref{riemann_stieltjes_form} and \eqref{szego_quadrature}, the above shows that from the eigensystem of $\widetilde{\textbf{U}}^T$ we can obtain $(\lambda_k,\omega_k)$ such that
\begin{equation}
\label{quantum_szego_rule}
    \langle\psi_0|f(U)|\psi_0\rangle=\sum_{k\in[d]}\omega_kf(\lambda_k) , 
\end{equation}
for any Laurent polynomial $f$ of degree up to $d-1$.
In other words, $(\lambda_k,\omega_k)$ is a Szeg\"o quadrature rule.
An illustration of this is given in \cref{fig:random_laurent_poly}, the details of which are discussed in \cref{numerics}.
Note that
\begin{equation}
\label{weight_normalization}
    \sum_{k\in[d]}\omega_k=1,
\end{equation}
since this sum is equal to the norm of the first row of the matrix of eigenvectors of $\widetilde{\textbf{U}}^T$.

\subsection{Beyond Laurent polynomials}
\label{beyond_laurent}

The above result on its own may not be too surprising, since in assembling $\textbf{U},\textbf{S}$ in \cref{quantum_iso_arnoldi}, we would have to obtain $\langle\psi_0|U^j|\psi_0\rangle$ for all $j=0,1,2,\ldots,d-1$.
Since those form the terms in
\begin{equation}
\label{direct_Laurent_series}
    \langle\psi_0|f(U)|\psi_0\rangle=\sum_{j=-d+1}^{d-1}\alpha_j\langle\psi_0|U^j|\psi_0\rangle
\end{equation}
for non-negative $j$, and those terms with negative $j$ are the complex conjugates of their positive-$j$ counterparts, we could evaluate the above expression directly using the entries in $\textbf{S}$.

The real power of the quantum Szeg\"o quadrature rule \eqref{quantum_szego_rule} comes when we use it to approximate $\langle\psi_0|f(U)|\psi_0\rangle$ for $f$ that is not a Laurent polynomial of degree $d-1$.
In this case the quadrature rule will no longer be exact, but can still yield an approximation whose accuracy improves as $d$ is increased.
This may be partly understood if $f$ is well-approximated by a Laurent polynomial: in that case, the quadrature rule will capture the Laurent polynomial component up to degree $d-1$ exactly, leaving only a residual to be approximated.
Since this holds for any Laurent polynomial $p$ up to degree $d-1$ approximating $f$, the error from QSQ will be determined by the least error from any such $p$, and unlike methods that explicitly construct such polynomial approximations, QSQ does not require knowing what the polynomial is.

To make this more precise, let $p^*$ be the optimal approximation of $f$ by a Laurent series of degree $d-1$.
We need to carefully define what we mean by optimal.
For example, a definition that ``knows'' about the exact form of $I$, like
\begin{equation}
    \argmin_{p\in\mathds{L}_{d-1}}|I(f)-I(p)| , 
\end{equation}
where $I(\cdot)$ is the integral \eqref{riemann_stieltjes_form} and $\mathds{L}_{d-1}$ denotes the set of degree-$(d-1)$ Laurent polynomials, is trivial since  as long as $I(p)\neq0$ we could obtain ${|I(f)-I(p)|=0}$ just by rescaling $p$ by some complex scalar.
Instead, we should consider a problem-agnostic optimality condition, in other words
\begin{equation}
    p^*=\argmin_{p\in\mathds{L}_{d-1}}\,\max_{z\in\mathds{T}}|f(z)-p(z)|,
\end{equation}
the minmax error over the complex unit circle $\mathds{T}$ (which contains the unknown spectrum of $U$).
Let $r$ be the residual
\begin{equation}
\label{r_def}
    r\coloneqq f-p^*.
\end{equation}

Let $R$ denote the quadrature rule (of degree $d$) as in \eqref{szego_quadrature}.
The error obtained from approximating $I(f)$ by $I(p^*)$ using the data in $\textbf{S}$ is
\begin{equation}
\label{laurent_error_theory}
    |I(p^*) - I(f)| = |I(r)|.
\end{equation}
Note that we obtain the same result with ${|R(p^*) - I(f)|}$, given that $R(p^*)=I(p^*)$ since $p^*$ is a Laurent polynomial of degree $d-1$.
The best upper bound for this error that is agnostic to $I$ is, by \eqref{riemann_stieltjes_form},
\begin{equation}
\begin{split}
    |I(r)|&=\left|\sum_{j\in[N]}|\gamma_j|^2r(\zeta_j)\right|\\
    &\le\sum_{j\in[N]}|\gamma_j|^2\epsilon=\epsilon,
\end{split}
\end{equation}
where, by \eqref{r_def},
\begin{equation}
    \epsilon\coloneqq\max_{z\in\mathds{T}}|f(z)-p^*(z)|=\min_{p\in\mathds{L}_{d-1}}\,\max_{z\in\mathds{T}}|f(z)-p(z)|.
\end{equation}

On the other hand, the error of the quadrature rule itself is
\begin{equation}
\label{rule_error_theory}
    |R(f) - I(f)| = |R(r) - I(r)|,
\end{equation}
again using $R(p^*)=I(p^*)$.
The best upper bound for this error that is agnostic to $I$ is
\begin{equation}
\begin{split}
    |R(r) - I(r)|&\le|R(r)| + |I(r)|\\
    &\le\left|\sum_{k\in[d]}\omega_kr(\lambda_k)\right| + \epsilon\\
    &\le\sum_{k\in[d]}\omega_k\epsilon + \epsilon = 2\epsilon,
\end{split}
\end{equation}
where the last step follows by \eqref{weight_normalization}.
In other words, if there exists a degree-$d$ Laurent polynomial approximation of $f$ whose error can be bounded by $\epsilon$, then QSQ with Krylov dimension $d$ will have error at most $2\epsilon$.

These are the optimal bounds one could calculate without knowing the spectral decomposition of $U$, and without making any further assumptions.
In practice for a specific $U$ and $|\psi_0\rangle$, the performance of either the optimal polynomial $p^*$ \eqref{laurent_error_theory} or the quadrature rule \eqref{rule_error_theory} could be better than the bound~--- either one could be zero by luck, and this occurrence would have no bearing on the performance of the other.
However, if the quadrature rule still captures some approximation of $I$ for the residual $r$, then the QSQ error \eqref{rule_error_theory} can outperform that of the optimal polynomial: we show an example of this in \cref{numerics}.
On the other hand, if the error $|I(r)|$ of the optimal polynomial is significantly smaller than $\epsilon$, then that really is by happenstance alone.

\section{Noise and regularization}
\label{stable_alg}

In order to develop an understanding of the performance of QSQ in practice, we must first discuss the impact of noise.
When the quantum subroutine (line 4) in \cref{quantum_iso_arnoldi} is evaluated on a quantum computer, the resulting matrices $(\textbf{U}', \textbf{S}')$ will have some errors with respect to their ideal counterparts $(\textbf{U}, \textbf{S})$:
\begin{equation}
    (\textbf{U}', \textbf{S}') = (\textbf{U}+\Delta\textbf{U}, \textbf{S}+\Delta\textbf{S}).
\end{equation}
At a minimum, the Heisenberg limit specifies that for a total runtime of $T$, the matrix elements must have errors scaling as $1/T$; if repeated sampling is used to estimate the matrix elements then instead the errors will scale as $1/\sqrt{T}$.
A full end-to-end analysis of the impact of those errors on the approximation of \eqref{qsq_output}, \emph{a la}~\cite{epperly2021subspacediagonalization, kirby2024analysis}, is beyond the scope of the present work, but we will at least discuss how the algorithm itself should be modified in the presence of noise and give numerical illustrations.

First, as discussed in detail in~\cite{kirby2024analysis}, the condition number of the ideal Gram matrix $\textbf{S}$ grows exponentially with $d$.
Since the maximum eigenvalue of $\textbf{S}$ grows linearly with $d$, this ill-conditioning is due to exponential decay of least eigenvalue with $d$.
In the presence of noise, this means that $\textbf{S}'$ may not even be positive semidefinite, so we must regularize $\textbf{S}'$ in order to use it.
In our case, the Gram-Schmidt orthonormalization in line 11 of \cref{quantum_iso_arnoldi} will fail if $\textbf{S}'$ is ill-conditioned.

When applying the quantum Krylov algorithm to the approximation of low-lying energies, the generalized eigenvalue problem is regularized by so-called \emph{eigenvalue thresholding}, i.e., projecting out eigenspaces of $\textbf{S}$ whose eigenvalues are below some specified positive threshold~\cite{epperly2021subspacediagonalization,klymko2022realtime,kirby2023exactefficient,kirby2024analysis,yoshioka2024diagonalization}.
In our case, eigenvalue thresholding turns out to fail.
The reason is that QSQ uses all eigenvalues and eigenvectors of $\widetilde{\textbf{U}}^T$ (see \cref{main_thm}), so while the analyses~\cite{epperly2021subspacediagonalization,kirby2024analysis} for low energies rely on the fact that projecting out some of the dimensions in the Krylov space turns out not to perturb the lowest-energy vector in that space too badly, in our case projecting out dimensions decreases the number of nodes and weights that we end up with.
Hence it cannot yield a Szeg\"o quadrature rule, and in practice it turns out to badly fail at approximating \eqref{qsq_output} at all.

Instead, we employ a variant of Tikhonov regularization, which has also been used in the context of quantum Krylov low-lying energy approximation~\cite{zhang2023measurementefficient}.
In other words, for some $\eta>0$ and $\lambda_\text{min}(\textbf{S}')$ the smallest eigenvalue of $\textbf{S}'$, we replace
\begin{equation}
    \textbf{S}'\rightarrow\widetilde{\textbf{S}}\coloneqq
    \begin{cases}
        \textbf{S}'-\lambda_\text{min}(\textbf{S}')+\eta\qquad&\text{if }\lambda_\text{min}(\textbf{S}')<\eta,\\
        \textbf{S}'\qquad&\text{if }\lambda_\text{min}(\textbf{S}')\ge\eta,
    \end{cases}
\end{equation}
which has smallest eigenvalue at least $\eta$, so is positive semidefinite by construction.
$\textbf{U}'$ is left unchanged, unlike in eigenvalue thresholding.

A second problem to resolve is the fact that while in the ideal case normalizing the last column of $\widetilde{\textbf{U}}$ makes it unitary, that is no longer sufficient once noise is included.
A convenient solution is to instead replace $\widetilde{\textbf{U}}$ by the closest unitary to it.
This is obtained by calculating the singular value decomposition $\widetilde{\textbf{U}}=PDQ^\dagger$ for unitary $P, Q$ and positive semidefinite diagonal $D$.
The closest unitary matrix to $\widetilde{\textbf{U}}$ in both Frobenius and $\ell_2$ norm distance is then $PQ^\dagger$~\cite{kahan2011nearest}.
Even better, we can show that in the limit of zero noise, this is equivalent to normalizing the last column of $\widetilde{\textbf{U}}$:
\begin{lemma}
    Let $\widetilde{\textbf{U}}$ be upper-Hessenberg with orthogonal columns that are all normalized except for the last, $\widetilde{\textbf{U}}_{,d-1}$.
    The closest unitary $PQ^\dagger$, obtained from the singular value decomposition $\widetilde{\textbf{U}}=PDQ^\dagger$, is identical to the matrix obtained from $\widetilde{\textbf{U}}$ by normalizing its last column, $\widetilde{\textbf{U}}_{,d-1}$.
\end{lemma}
\begin{proof}
    The diagonal entries of $D$ are the singular values of $\widetilde{\textbf{U}}$, i.e., the square-roots of the eigenvalues of $\widetilde{\textbf{U}}^\dagger\widetilde{\textbf{U}}$.
    Since the columns of $\widetilde{\textbf{U}}$ are orthogonal, $\widetilde{\textbf{U}}^\dagger\widetilde{\textbf{U}}$ is diagonal with $i$th diagonal entry $\|\widetilde{\textbf{U}}_{,i}\|^2$, i.e., the length of the $i$th column.
    These are all $1$ except for the last, $\|\widetilde{\textbf{U}}_{,d-1}\|^2$, and they also correspond to the squares of the diagonal entries in $D$ up to a permutation.
    Hence $D$ has a single non-unity diagonal entry given by $\|\widetilde{\textbf{U}}_{,d-1}\|$.
    Normalizing the last column of $\widetilde{\textbf{U}}$ is therefore equivalent to replacing that entry with $1$, which also yields our desired transformation to the nearest unitary,
    \begin{equation}
        \widetilde{\textbf{U}}=PDQ^\dagger~\rightarrow~PIQ^\dagger=PQ^\dagger,
    \end{equation}
    where $I$ is the $d\times d$ identity.
\end{proof}

The final modification we will make in order to stabilize our algorithm is the method of orthogonalizing the Krylov basis.
While the Gram-Schmidt orthonormalization as in \cref{quantum_iso_arnoldi} is useful for formally proving the results in \cref{algorithm}, it is not a numerically stable method for orthogonalizing the basis.
Instead, once we have regularized to obtain $\widetilde{\textbf{S}}$, we can orthonormalize as
\begin{equation}
\label{stable_orthogonalization}
    \widetilde{\textbf{U}}=\widetilde{\textbf{S}}^{-1/2}\textbf{U}'\widetilde{\textbf{S}}^{-1/2},
\end{equation}
which is stable.
This yields $\widetilde{\textbf{U}}$ that is related to the $\widetilde{\textbf{U}}$ of \cref{quantum_iso_arnoldi} (line 11) by a unitary transformation.
The isometric step (line 12 in \cref{quantum_iso_arnoldi}) is replaced by the mapping to closest unitary procedure above, which is basis independent, so after that step the new $\widetilde{\textbf{U}}^T$ has the same eigenvalues as that of \cref{quantum_iso_arnoldi} (line 13), which form the nodes of the Szeg\"o quadrature rule.

However, its eigenvectors are different, since it is expressed in a different basis.
Recall that the weights of the quadrature rule are given by $|v_{k,0}|^2$ \eqref{weights_formula}, i.e., the square-magnitudes of the zeroth entries of the eigenvectors in the basis corresponding to Gram-Schmidt orthonormalization.
We now have new eigenvectors $\tilde{v}_k$ of $\widetilde{\textbf{U}}^T$ in the new basis \eqref{stable_orthogonalization}.
We would like to find the weights from these eigenvectors.

We do not have direct access to the basis transformation, but we do have the transformation back to the original, nonorthogonal basis: from \eqref{stable_orthogonalization}, the $k$th eigenvector is $\widetilde{\textbf{S}}^{1/2}\tilde{v}_k$ when expressed in the nonorthogonal basis.
We also know that Gram-Schmidt leaves zeroth entries intact, and $v_k$ is obtained by applying Gram-Schmidt to the eigenvectors in the nonorthogonal basis, so the zeroth entry of $v_k$ is unchanged from its value in the nonorthogonal basis:
\begin{equation}
\label{nonortho_to_ortho}
    v_{k,0}=[\widetilde{\textbf{S}}^{1/2}\tilde{v}_k]_0,
\end{equation}
where $[\widetilde{\textbf{S}}^{1/2}\tilde{v}_k]_0$ denotes the zeroth entry of the vector $\widetilde{\textbf{S}}^{1/2}\tilde{v}_k$.
The right-hand side in \eqref{nonortho_to_ortho} is the inner product of the zeroth row of $\widetilde{\textbf{S}}^{1/2}$, which we denote $\vec{s}_0$, with $\tilde{v}_k$, so we obtain
\begin{equation}
    v_{k,0}=(\vec{s}_0,\tilde{v}_k).
\end{equation}
Hence squaring these inner products for each $k$ yields the weights:
\begin{equation}
    |v_{k,0}|^2=|(\vec{s}_0,\tilde{v}_k)|^2.
\end{equation}

With these changes incorporated, we end up with a new, stable version of quantum Szeg\"o quadrature.
We present this by first wrapping the regularization and projection-to-unitary procedure into one subroutine:

\begin{widetext}
\begin{myalgorithm}[``Regularization and projection-to-unitary"]~
\label{regularization_alg}
\begin{algorithmic}[1]

    \State\textbf{input} noisy quantum Krylov matrices $(\textbf{U}',\textbf{S}')$

    \State Calculate $\lambda_\text{min}(\textbf{S}')$ \Comment{Tikhonov regularization}
    
    \State Choose $\eta>0$

    \If{$\lambda_\text{min}(\textbf{S}')<\eta$}
    
    \State $\widetilde{\textbf{S}}\gets\textbf{S}'-\lambda_\text{min}(\textbf{S}')+\eta$

    \Else

    \State $\widetilde{\textbf{S}}\gets\textbf{S}'$

    \EndIf

    \State $\widetilde{\textbf{U}} \gets\widetilde{\textbf{S}}^{-1/2}\textbf{U}'\widetilde{\textbf{S}}^{-1/2}$ \Comment{orthonormalization}

    \State $PDQ^\dagger=\widetilde{\textbf{U}}$ \Comment{singular value decomposition}

    \State $\widetilde{\textbf{U}} \gets PQ^\dagger$ \Comment{isometric Arnoldi}

    \State\textbf{return} $\widetilde{\textbf{U}}$
    
\end{algorithmic}
\end{myalgorithm}

The regularization step (lines 4-7) guarantees that the Gram matrix $\widetilde{\textbf{S}}$ has smallest eigenvalue of at least $\eta$.
Using this subroutine, we can now present the full stabilized quantum Szeg\"o quadrature algorithm:

\begin{myalgorithm}[``Quantum Szeg\"o Quadrature~--- stable version"]~
\label{noisy_quantum_iso_arnoldi}
\begin{algorithmic}[1]

    \State\textbf{input} $n$-qubit unitary operator $U$, Krylov dimension $d$, initial state $|\psi_0\rangle$

    \State $\textbf{X}_{0}\gets1$

    \For{$j=1,2,\ldots,d$} \Comment{on quantum (Krylov circuits)}

    \State Calculate $\textbf{X}'_{j}\approx\langle\psi_0|U^{j}|\psi_0\rangle\qquad$

    \State $\textbf{X}_{-j}\gets\overline{\textbf{X}_{j}}$

    \EndFor

    \For{$i,j=0,1,2,\ldots,d-1$} \Comment{build noisy Krylov matrices $(\textbf{U}',\textbf{S}')$}

    \State $\textbf{U}'_{ij}\gets\textbf{X}'_{j-i+1}$

    \State $\textbf{S}'_{ij}\gets\textbf{X}'_{j-i}$

    \EndFor

    \State $F \gets$ \textbf{\cref{regularization_alg}}

    \State $\widetilde{\textbf{U}} \gets F(\textbf{U}',\textbf{S}')$

    \State $\vec{\lambda} = \text{eigvals}\big(\widetilde{\textbf{U}}^T\big)$

    \State $V = \text{eigvecs}\big(\widetilde{\textbf{U}}^T\big)$

    \State $\vec{s}\gets\big[\widetilde{\textbf{S}}^{1/2}\big]_{0,}$ \Comment{zeroth row of $\widetilde{\textbf{S}}^{1/2}$}

    \State $\vec{\omega}\gets\vec{s}\,V$

    \State $\vec{\omega}\gets\vec{\omega}^{\,2}$ \Comment{entrywise square}

    \State\textbf{return} $\vec{\lambda},\vec{\omega}$
    
\end{algorithmic}
\end{myalgorithm}
\end{widetext}

The returned vectors are the nodes and corresponding weights of the Szeg\"o quadrature rule \eqref{szego_quadrature}.
The approximation in line 4 indicates the inclusion of finite-sample, algorithmic (e.g., Trotter), and device error in the estimates $(\textbf{U}',\textbf{S}')$.

\section{Numerics}
\label{numerics}

In order to construct classical numerical illustrations of the performance of QSQ, we must choose representative unitaries $U$ and states $|\psi_0\rangle$.
We choose for our $U$ a time evolution operator of a Heisenberg model on a $4\times3$ square lattice, whose Hamiltonian is
\begin{equation}
\label{heis_ham}
    H=\sum_{i}hZ_i + \sum_{\langle i,j\rangle}\left(j_1X_iX_j+j_2Y_iY_j+j_3Z_iZ_j\right),
\end{equation}
where $\langle i,j\rangle$ denotes neighboring pairs in the lattice.
We use a version of an XXZ model, namely ${h=j_1=j_2=1}$, $j_3=2$.
The time evolution $U$ is then defined
\begin{equation}
    U=e^{-iH\,\Delta t},
\end{equation}
with $\Delta t=\pi/\|H\|$.
For $|\psi_0\rangle$, we use the $6$-particle (half-filling) antiferromagnetic state, i.e.,
\begin{equation}
\label{psi_0_def}
    |\psi_0\rangle=|1010\ldots10\rangle
\end{equation}
where the qubits in the lattice are ordered $(0,0),(0,1),(0,2),(0,3),(1,0),\ldots$
This state carries $99.9\%$ of its probability in 165 energy eigenstates, and is entirely supported on 272 eigenstates, so at the Krylov dimensions we consider there will not be a risk of saturating the effective Hilbert space.

We begin by checking in \cref{fig:random_laurent_poly} that the analytic claims in \cref{algorithm} do in fact hold in practice.
This figure shows average errors from applying QSQ at a given Krylov dimension to random Laurent polynomials of a given degree.
We know from \cref{algorithm} that QSQ should be exact for degree-$d$ Laurent polynomials at Krylov dimension $d+1$, and this is exactly what we see in \cref{fig:random_laurent_poly}: for each degree $d$ we see the error ``jump'' down to the range of machine precision at Krylov dimension $d+1$.

We can also test the stability of QSQ in the presence of noise, as discussed in \cref{stable_alg}.
To illustrate this, we consider a degree-$5$ monomial $U^5$, which in the absence of noise we would expect to converge to exact in degree-$6$ QSQ.
Results of simulating QSQ with input data subject to Gaussian noise of varying strength are shown in \cref{fig:noisy_monomial}.
As the figure shows, once the degree required for convergence is reached, the output relative errors converge to values that scale approximately linearly with the input noise strength.
This indicates stability in the presence of noise, although a more in-depth numerical and end-to-end analytical study would be valuable and we defer that to future work.

\begin{figure}[t]
    \centering
    \includegraphics[width=\linewidth]{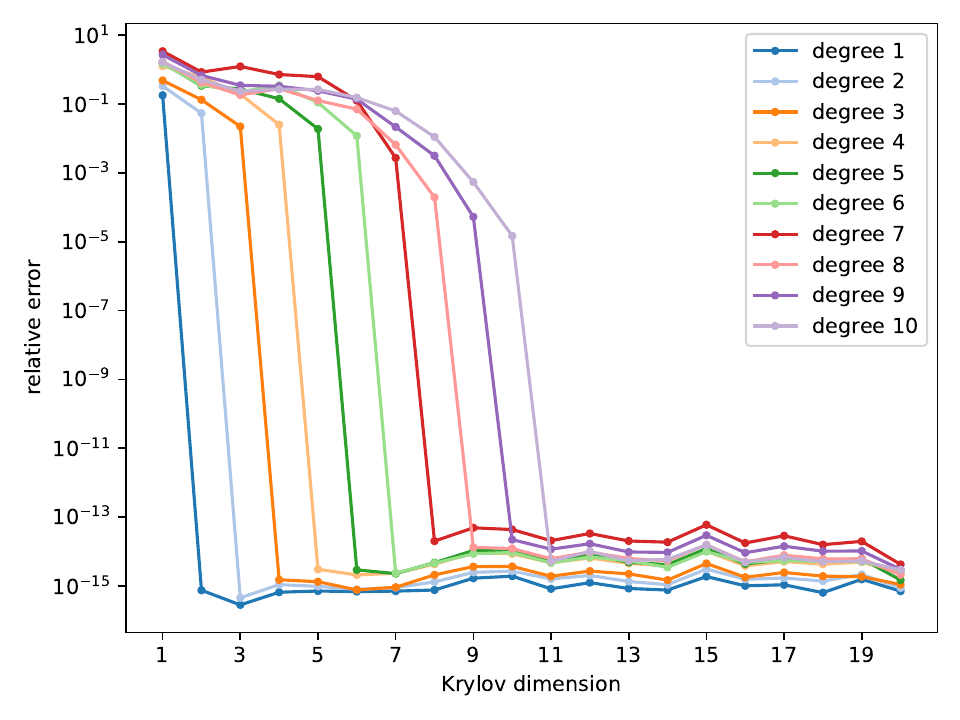}
    \caption{Each point is the mean of the relative errors in approximating \eqref{qsq_output} for 10 random Laurent polynomials of degree given in the legend, plotted against the Krylov dimension used for QSQ. The analysis in \cref{algorithm} shows that degree-$d$ Laurent polynomials should be captured exactly at Krylov dimension $d+1$, which is what we see. Left-to-right order of the curves corresponds to top-to-bottom order of the legend.}
    \label{fig:random_laurent_poly}
\end{figure}

\begin{figure}[t]
    \centering
    \includegraphics[width=\linewidth]{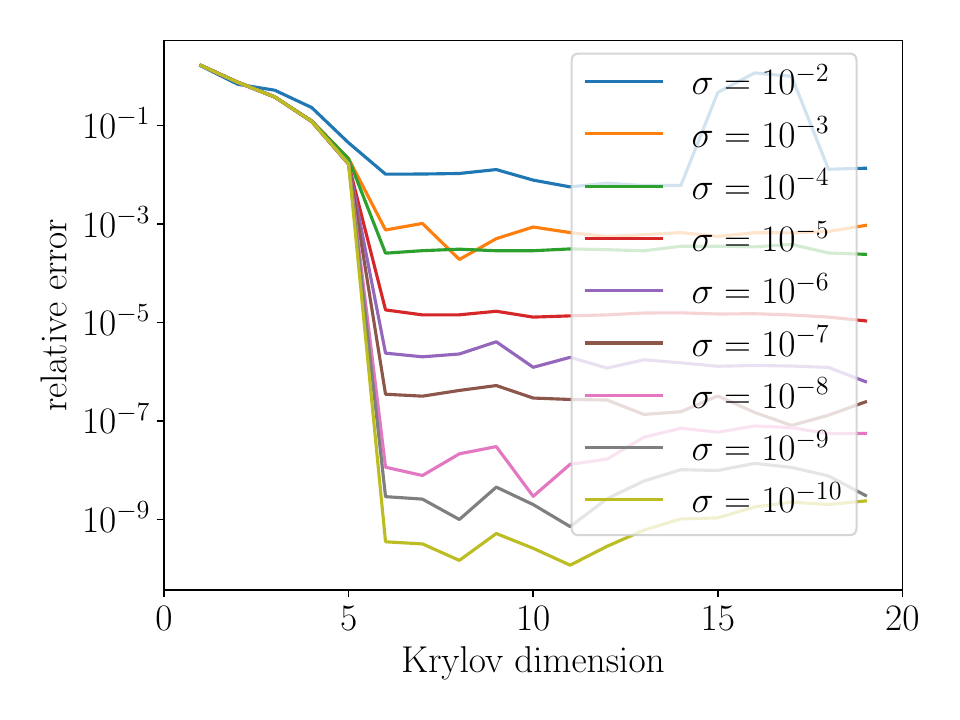}
    \caption{Performance of QSQ for $f(U)=U^5$ with input data subject to Gaussian noise of width $\sigma$. Note that the output relative errors scale approximately linearly with $\sigma$ once the degree required for convergence (6) is reached. Vertical order of the curves corresponds to vertical order of the legend.}
    \label{fig:noisy_monomial}
\end{figure}

Going beyond polynomials, we study three cases of functions that we can only hope to approximate using a quadrature rule, but for which we will see that the rule can still converge to good accuracy.
The first of these cases is a (unnormalized) Gibbs state, i.e., the function
\begin{equation}
\label{gibbs_state_fn}
    f(U)=U^{-i\beta/\Delta t}=e^{-\beta H}.
\end{equation}
With this function, the functional \eqref{qsq_output} may be of direct interest in the context of estimating the trace of \eqref{gibbs_state_fn}, which is the partition function.
In that case, $|\psi_0\rangle$ would be drawn from some distribution and the resulting quantities \eqref{qsq_output} would be used to form a trace estimator (e.g.~\cite{sirui2021,Ghanem2023robustextractionof,shen2024trace,summer2024,goh2024directestimationdensitystates}).
In some special cases (for example in the thermodynamic limit depending on the temperature) we might be able to use a single state $|\psi_0\rangle$ given a typicality condition~\cite{sugiura2012,sugiura2013}.

More generally, the functional \eqref{qsq_output_general} can be used to estimate, for example, the trace of a product $\hat{O}e^{-\beta H}$ for an observable $\hat{O}$, provided the action of $\hat{O}$ on $|\psi_0\rangle$ can be calculated classically.
Given this, we can take
\begin{equation}
    |\psi_1\rangle=\frac{\hat{O}|\psi_0\rangle}{\|\hat{O}|\psi_0\rangle\|},
\end{equation}
and the expression \eqref{qsq_output_general} becomes
\begin{equation}
    \|\hat{O}|\psi_0\rangle\|\langle\psi_1|f(U)|\psi_0\rangle=\langle\psi_0|\hat{O}e^{-\beta H}|\psi_0\rangle,
\end{equation}
which can then be used in a trace estimator as above to estimate $\text{Tr}(\hat{O}e^{-\beta H})$.

For the numerical illustration, we focus on the simpler form \eqref{qsq_output}, since as discussed above it can be used to construct the general form, so its accuracy will be representative of the general case accuracy.
In \cref{fig:gibbs_varying_beta}, we show the results of QSQ for estimating
\begin{equation}
\label{gibbs_exp}
    \langle\psi_0|e^{-\beta H}|\psi_0\rangle=\langle\psi_0|U^{-i\beta/\Delta t}|\psi_0\rangle , 
\end{equation}
as a function of Krylov dimension, for various values of $\beta$.
Setting aside variations in convergence at early iterations and once the accuracy approaches machine precision, these all show exponential decay at roughly the same rate, between $O(e^{-d})$ and $O(e^{-2d})$, which are the upper and lower dashed lines in the plot.

\begin{figure}[t]
    \centering
    \includegraphics[width=\linewidth]{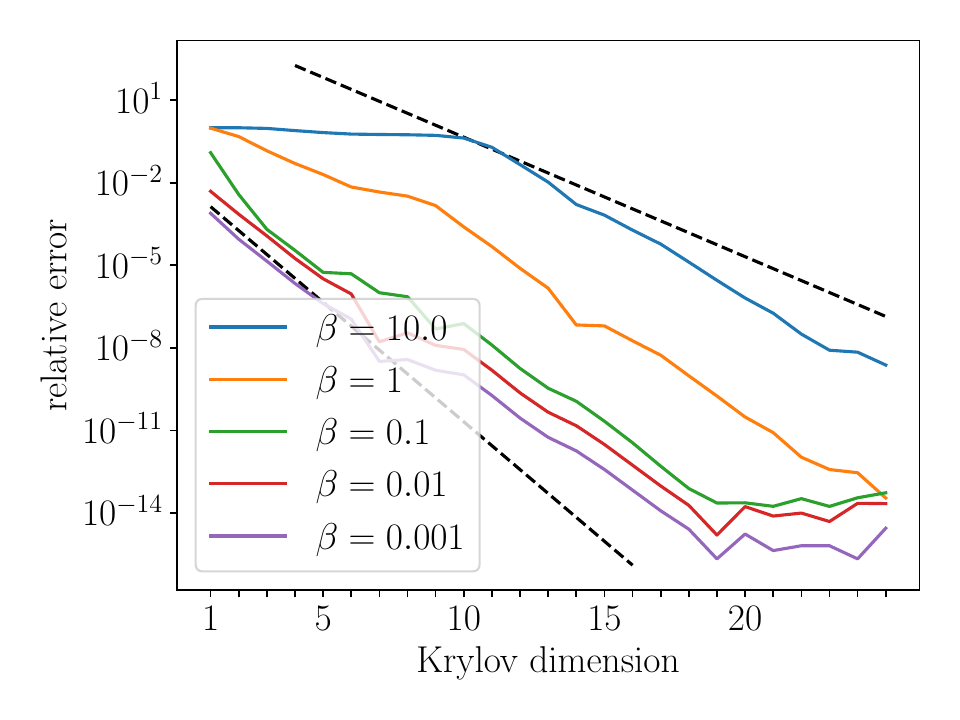}
    \caption{Relative errors in approximating the expectation value of the unnormalized Gibbs state as in \eqref{gibbs_exp}, as a function of Krylov dimension (degree of the quadrature rule). The upper and lower dashed lines represent $O(e^{-d})$ and $O(e^{-2d})$, respectively. Vertical order of the curves corresponds to vertical order of the legend.}
    \label{fig:gibbs_varying_beta}
\end{figure}

The Gibbs state is also a convenient example for motivating the use of QSQ in the beyond-polynomial setting, as opposed to some other post-processing one could perform using the same quantum data.
As discussed above, one might imagine approximating $f$ directly as in \eqref{direct_Laurent_series}, since the terms on the right-hand side are exactly the quantities we measure in the quantum part of the experiment.
The only question in that case is how to choose the coefficients $\alpha_j$ in order to best approximate our desired function.
We consider three variants of this approach:
\begin{enumerate}
    \item Express $e^{-\beta H}$ as a Fourier series in $H$ and truncate it at order $d-1$.
    \item An analytic series construction to approximate $f$ by a Laurent polynomial at each fixed order, with no requirement that lower-order approximations be partial sums of higher-order. Such a series with provably exponential convergence for $e^{-\beta H}$ can be obtained from~\cite[Lemma 37]{vanApeldoorn2020quantumsdpsolvers}, along with an upper bound on its error; see \cref{app:approximating gibbs state} for details. The resulting error bound is
    \begin{equation}
    \label{fixed_bound}
        \left\|p_d(U)-e^{-\beta H}\right\|\le4\exp\left(\frac{\beta\|H\|}{\gamma}-\frac{d}{2}(1-\gamma)\right) , 
    \end{equation}
    where $p_d(U)$ is the degree-$d$ Laurent polynomial and $\gamma$ is a free parameter in $(0,1)$.
    \item Best of all, we can directly minimize ${\|f(U)-\sum_{j=-d+1}^{d-1}\alpha_jU^j\|}$ over $\vec{\alpha}$. This has a classically-calculable optimum for fixed-size problems, although finding it is not scalable with respect to problem size since it requires manipulations of the full matrix or the full spectral decomposition of $U$. However, it illustrates the best possible scenario for Laurent polynomial approximation of $f$ in the specific model we are studying.
\end{enumerate}

We compare the above three approaches to QSQ at $\beta=1$.
The results are shown in \cref{fig:gibbs_comparison}.
The truncated Fourier series converges linearly (although on the log-log scale this is not easy to distinguish from some other polynomial), and does not reach below relative error 1 until dimension 417.
Its large error is due to the fact that, although the Fourier series approximates $e^{-\beta x}$ well in the middle of the spectral range, it has large fluctuations near the ends of the range due to the steepness of the exponential curve.
The severity of this failure is set by the temperature, since that controls the decay of $e^{-\beta x}$, so for higher temperatures we could expect the Fourier series to perform better overall, but still to converge only linearly with the Krylov dimension.

The bound on the error for the fixed Laurent approximation does yield exponential convergence asymptotically. 
At each degree, we choose the value of $\gamma\in(0,1)$ that minimizes the error bound \eqref{fixed_bound}, and since the bound is in spectral norm we obtain a relative error by normalizing with a factor of $\|e^{-\beta H}\|^{-1}$.
The resulting improvement over the Fourier series is possible because higher-degree approximations are not constrained to contain lower-degree approximations as partial sums.
However, as \cref{fig:gibbs_comparison} illustrates, we do not see the exponential convergence start to kick in until degrees around one hundred.
We can understand this since we optimize the bound \eqref{fixed_bound} over $\gamma$, and the optimal $\gamma$ will not begin to move away from $1$ until $d$ surpasses $\sim2\beta\|H\|$, which is $66$ in the current example.

One might reasonably object that this bound is likely not tight and the actual fixed Laurent series might outperform it.
Instead of plotting that, we plot a lower bound for the error of any Laurent series approximation by plotting the error of the optimal Laurent series at each degree.
In this case, we optimize the approximation over the eigenvalues of $U$ itself, yielding an even better approximation than the bound discussed in the last few paragraphs of \cref{algorithm}.
This approximation is unfeasible to obtain in practice since it is as hard as calculating $f(U)$ directly, but it provides a lower bound on the error of any approximation relying on an explicit series expansion \eqref{direct_Laurent_series} using the data in $\textbf{S}$, unless that expansion additionally uses knowledge of the decomposition of $|\psi_0\rangle$ in the eigenbasis of $U$.
While it far outperforms the two prior approaches, its error is still several orders of magnitude worse than that of QSQ after the first few dimensions, requiring about 15 more dimensions in order to reach machine precision.

We may understand this by repeating the points noted at the end of \cref{algorithm}, namely that QSQ of dimension $d$ not only captures the optimal Laurent polynomial approximation of degree $d-1$, it provides some quadrature approximation of the residual.
As noted in \cref{algorithm}, if the residual is upper bounded in magnitude by $b$, then we can upper bound the error of QSQ by $2b$, but the fact that the red QSQ curve so dramatically outperforms the green optimal Laurent curve in \cref{fig:gibbs_comparison} illustrates that the quadrature approximation of the residual is far better than the worst case in this example.
To summarize, QSQ at worst captures the error bound of the optimal Laurent approximation (up to a factor of two) without needing to actually find that approximation, and in practice it can substantially improve even on that due to the quadrature approximation of the residual.

\begin{figure}[t]
    \centering
    \includegraphics[width=\linewidth]{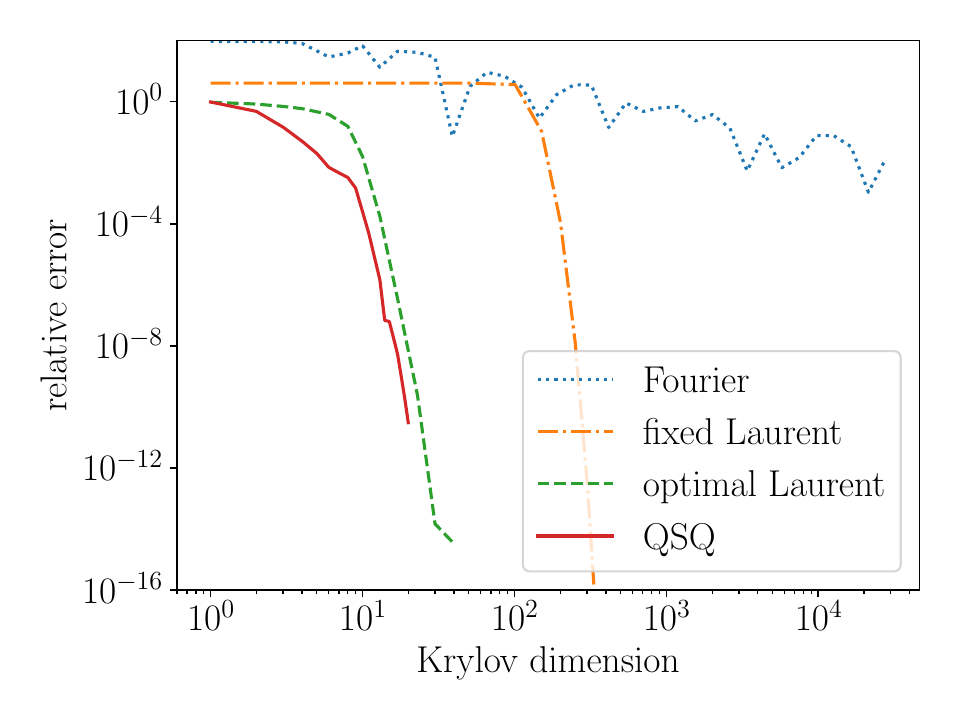}
    \caption{Comparing QSQ against several other methods that could be used to approximate the Gibbs state expectation value \eqref{gibbs_exp} using the same quantum data $(\textbf{U},\textbf{S})$ (see lines 3 and 4 in \cref{quantum_iso_arnoldi}). ``Fourier'' refers to expressing the Gibbs state as a Fourier series up to order $d-1$, evaluated using the data in $\textbf{S}$. ``Fixed Laurent'' refers to a specific construction for a Laurent series approximation of \eqref{gibbs_state_fn}. ``Optimal Laurent'' refers to optimizing over all possible Laurent series approximations of $f(U)$ at each degree, which has a classically-calculable optimum. See main text for a discussion of the results.}
    \label{fig:gibbs_comparison}
\end{figure}

The final numerical example we consider is the retarded Green's function in the frequency domain,
\begin{equation}
    G^R(\omega)=\langle\psi_0|\frac{1}{H-\omega-i\chi}|\psi_0\rangle,
\end{equation}
where $\omega$ is the frequency (expressed in units of energy) and $\chi$ is a regularization parameter that shifts the poles away from the real axis and thus controls the smoothness of the function.
In the following numerics, we set $\chi=0.1$.
\cref{fig:gf_fig} shows results of approximating $G^R(\omega)$ using QSQ, alongside the exact results.
These plots show how, as the degree of the quadrature rule is increased, the locations of the strong peaks (corresponding to poles with significant amplitude in $|\psi_0\rangle$) are refined and more of the detailed features are resolved.

To quantify the improvement in accuracy with degree we can also take the $l_1$-norm of the error, i.e., the pointwise differences in \cref{fig:gf_fig}, over the spectral range, which in our case is contained in $[-60,20]$, and plot that against the degree $d$ of the quadrature rule.
The results are shown in \cref{fig:gf_error}, illustrating that the convergence is approximately inverse-linear up until degree $\sim100$.
Beyond that point there is an acceleration due to the dimension of the Krylov subspace beginning to saturate the effective Hilbert space, i.e., the size of the support of $|\psi_0\rangle$ \eqref{psi_0_def} in the energy spectrum.

\begin{figure*}[ht]
    \centering
    \includegraphics[width=\linewidth]{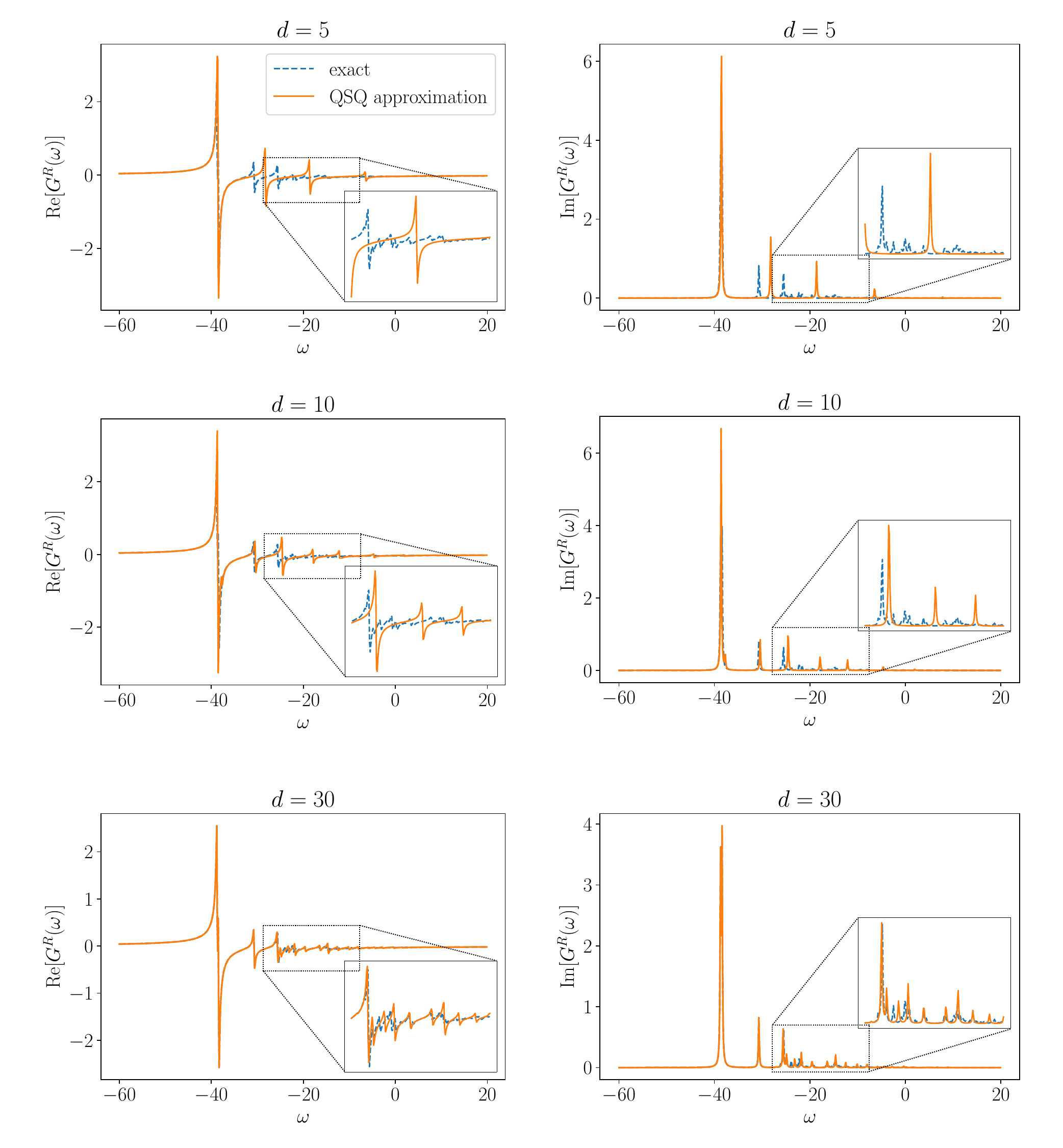}
    \caption{Approximation of the retarded Green's function using quantum Szeg\"o quadrature. The left column shows real parts and the right column shows imaginary parts, with the solid curves showing the approximation and the dashed curves showing the exact values. The rows correspond to Krylov dimensions of 5, 10, and 30, and hence quadrature rules of those degrees. The energies with significant amplitude in $|\psi_0\rangle$ \eqref{psi_0_def} are contained in $[-60,20]$, so all features of the Green's function are reflected in the above.}
    \label{fig:gf_fig}
\end{figure*}

\begin{figure}[ht]
    \centering
    \includegraphics[width=\linewidth]{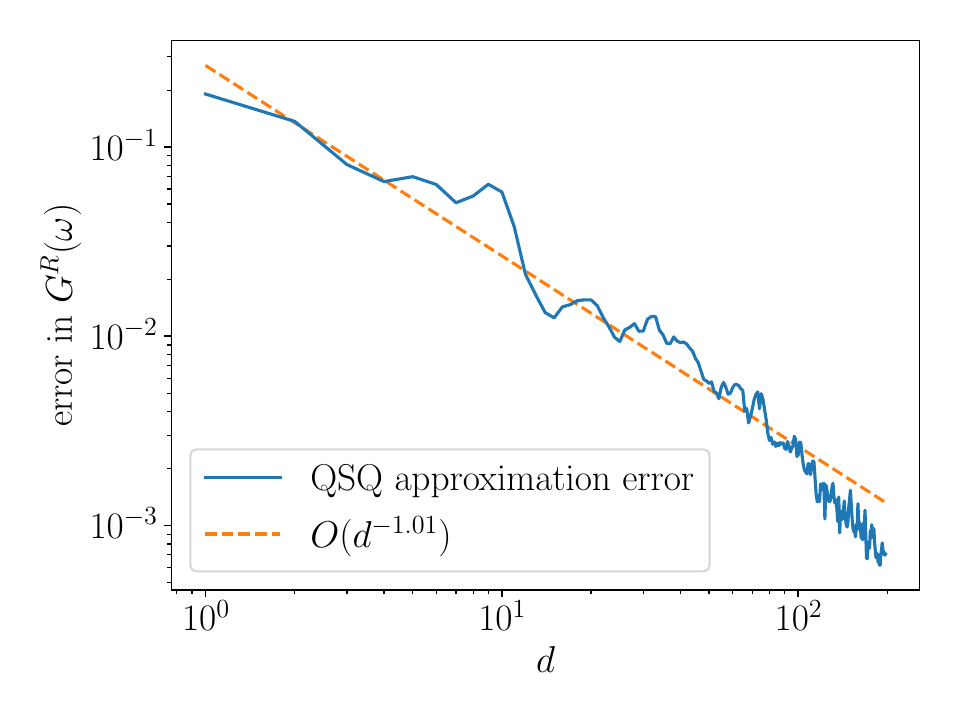}
    \caption{$l_1$-error of the QSQ approximation of $G^R(\omega)$ over the range $[-60,20]$, as a function of degree $d$ of the quadrature rule. The solid curve shows the actual data, while the dashed curve is the optimal linear fit to the log-log data (weighted to normalize the logarithmic point density). The legend gives the leading order of the corresponding power law. The acceleration of the decay of the data above $d\sim100$ is due to the dimension of the Krylov subspace beginning to saturate the effective Hilbert space.}
    \label{fig:gf_error}
\end{figure}

\section{Discussion}
\label{discussion}

We have presented a quantum algorithm, quantum Szeg\"o quadrature or QSQ, for approximating matrix elements of a unitary operator via a Szeg\"o quadrature rule constructed from quantum data.
QSQ can be compared to prior algorithms such as the quantum eigenvalue transformation of unitaries (QET-U)~\cite{dong2022groundstate}, as discussed in the introduction.
The main differences are, first, that QSQ yields a quadrature rule that can be used to approximate \eqref{qsq_output_general} for any function $f$ after the fact, while QET-U prepares a block-encoded approximation of $f(U)|\psi_0\rangle$ for a fixed $f$.
While in the latter case we can subsequently perform any quantum processing we like on the block-encoded state, QSQ's ability to approximate the matrix element for many different $f$ using the same quantum data is powerful, as we illustrated in the Green's function example in \cref{numerics}.

A second advantage is that QET-U requires constructing an explicit Laurent series approximation of $f$, while QSQ only requires the ability to evaluate $f$ at the nodes of the quadrature rule.
High-precision polynomial approximations of many functions of interest are indeed known, but the fact of having to construct them classically \emph{a priori} additionally means that they must be accurate over the entire spectral range and cannot depend on the specific distribution of eigenvalues.
QSQ, on the other hand, can be seen as constructing its approximation on the fly from quantum data about the unitary.
Hence it can lead to lower errors than any possible direct construction of a polynomial approximation.

To flesh out this idea further, at the end of \cref{algorithm}, we noted that the data $\textbf{X}$ collected from the quantum computer in \cref{quantum_iso_arnoldi} are approximations of
\begin{align}
\label{quantum_data}
\textbf{X}_j&=\langle\psi_0|U^j|\psi_0\rangle,& j&=-d,-d+1,\ldots,d.
\end{align}
A natural observation is therefore that instead of using the quadrature rule to approximate \eqref{qsq_output}, we could simply approximate $f$ by a Laurent polynomial of degree $d-1$ and then express \eqref{qsq_output} as the appropriate linear combination of terms in \eqref{quantum_data}.
While this is true, using QSQ can yield lower errors as shown in the numerical example \cref{fig:gibbs_comparison}, and it is guaranteed to produce at worst twice the error of the optimal polynomial approximation of degree $d-1$, as shown in \cref{beyond_laurent}.

A second alternative version of the algorithm would be to skip the calculation of eigenvalues and eigenvectors of $\widetilde{\textbf{U}}^T$ as in \cref{main_thm}, and instead simply return
\begin{equation}
    \big[f\big(\widetilde{\textbf{U}}^T\big)\big]_{0,0},
\end{equation}
the $(0,0)$th matrix element of $f\big(\widetilde{\textbf{U}}^T\big)$, as the approximation of \eqref{qsq_output}.
The justification for this is that if $\Lambda$ is the diagonal matrix of eigenvalues of $\widetilde{\textbf{U}}^T$, $V$ the corresponding matrix of eigenvectors, and $f$ has a Laurent polynomial expansion, then
\begin{equation}
\label{alternative_form}
\begin{split}
    \big[f\big(\widetilde{\textbf{U}}^T\big)\big]_{0,0}&=\big[f\big(V\Lambda V^\dagger\big)\big]_{0,0}\\
    &=\big[Vf\big(\Lambda\big)V^\dagger\big]_{0,0}\\
    &=\sum_{k=0}^{d-1}|v_{k,0}|^2f(\lambda_k).
\end{split}
\end{equation}
If $f$ does not have a Laurent polynomial expansion, then since $\widetilde{\textbf{U}}$ is always diagonalizable, the extension of the scalar function $f$ to a matrix function is defined by
\begin{equation}
    f\big(V\Lambda V^\dagger\big)\coloneqq Vf\big(\Lambda\big)V^\dagger , 
\end{equation}
so the second step above still holds.
The reader may recall that the final expression in \eqref{alternative_form} is exactly our quadrature rule, as given by \eqref{weights_formula} and \eqref{quantum_szego_rule}.
Hence this is simply an alternative final step for the algorithm, albeit one that obscures the interpretation as a quadrature rule, which is what allowed us to prove convergence for Laurent polynomials.

~
\begin{acknowledgments}

The authors thank Friederike Metz, Kunal Sharma, Minh Tran, Javier Robledo Moreno, and Mirko Amico for helpful discussions.  This research was supported by the U.S. Department of Energy (DOE) under Contract No. DEAC02-05CH11231, through the National Energy Research Scientific Computing Center (NERSC), an Office
of Science User Facility located at Lawrence Berkeley National Laboratory.

\end{acknowledgments}

\bibliographystyle{ieeetr}
\bibliography{references}

\begin{thebibliography}{10}

\bibitem{poulin2009sampling}
D.~Poulin and P.~Wocjan, ``Sampling from the thermal quantum gibbs state and evaluating partition functions with a quantum computer,'' {\em Phys. Rev. Lett.}, vol.~103, p.~220502, Nov 2009.

\bibitem{poulin2009preparing}
D.~Poulin and P.~Wocjan, ``Preparing ground states of quantum many-body systems on a quantum computer,'' {\em Phys. Rev. Lett.}, vol.~102, p.~130503, Apr 2009.

\bibitem{hhl2009}
A.~W. Harrow, A.~Hassidim, and S.~Lloyd, ``Quantum algorithm for linear systems of equations,'' {\em Phys. Rev. Lett.}, vol.~103, p.~150502, Oct 2009.

\bibitem{childs2010quantumwalk}
A.~M. Childs, ``On the relationship between continuous- and discrete-time quantum walk,'' {\em Communications in Mathematical Physics}, vol.~294, no.~2, pp.~581--603, 2010.

\bibitem{childs2012lcu}
A.~M. Childs and N.~Wiebe, ``Hamiltonian simulation using linear combinations of unitary operations,'' {\em Quantum Information and Computation}, vol.~12, no.~11-12, pp.~901--924, 2012.

\bibitem{childs2017quantum}
A.~M. Childs, R.~Kothari, and R.~D. Somma, ``Quantum algorithm for systems of linear equations with exponentially improved dependence on precision,'' {\em SIAM Journal on Computing}, vol.~46, no.~6, pp.~1920--1950, 2017.

\bibitem{chowdhury2017quantum}
R.~D.~S. Anirban~Chowdhury, ``Quantum algorithms for gibbs-sampling and hitting-time estimation,'' {\em Quantum Information and Computation}, vol.~17, no.~1, pp.~0041--0064, 2017.

\bibitem{vanApeldoorn2020quantumsdpsolvers}
J.~van Apeldoorn, A.~Gily{\'{e}}n, S.~Gribling, and R.~de~Wolf, ``Quantum {SDP}-{S}olvers: {B}etter upper and lower bounds,'' {\em {Quantum}}, vol.~4, p.~230, Feb. 2020.

\bibitem{ge2019faster}
Y.~Ge, J.~Tura, and J.~I. Cirac, ``Faster ground state preparation and high-precision ground energy estimation with fewer qubits,'' {\em Journal of Mathematical Physics}, vol.~60, p.~022202, 02 2019.

\bibitem{low2017signalprocessing}
G.~H. Low and I.~L. Chuang, ``Optimal hamiltonian simulation by quantum signal processing,'' {\em Phys. Rev. Lett.}, vol.~118, p.~010501, Jan 2017.

\bibitem{low2019qubitization}
G.~H. Low and I.~L. Chuang, ``Hamiltonian {S}imulation by {Q}ubitization,'' {\em {Quantum}}, vol.~3, p.~163, July 2019.

\bibitem{gilyen2019qsvt}
A.~Gily\'{e}n, Y.~Su, G.~H. Low, and N.~Wiebe, ``Quantum singular value transformation and beyond: Exponential improvements for quantum matrix arithmetics,'' in {\em Proceedings of the 51st Annual ACM SIGACT Symposium on Theory of Computing}, (New York, NY, USA), p.~193–204, Association for Computing Machinery, 2019.

\bibitem{dong2022groundstate}
Y.~Dong, L.~Lin, and Y.~Tong, ``Ground-state preparation and energy estimation on early fault-tolerant quantum computers via quantum eigenvalue transformation of unitary matrices,'' {\em PRX Quantum}, vol.~3, p.~040305, Oct 2022.

\bibitem{mcclean2017subspace}
J.~R. McClean, M.~E. Kimchi-Schwartz, J.~Carter, and W.~A. de~Jong, ``Hybrid quantum-classical hierarchy for mitigation of decoherence and determination of excited states,'' {\em Phys. Rev. A}, vol.~95, p.~042308, Apr 2017.

\bibitem{colless2018computation}
J.~I. Colless, V.~V. Ramasesh, D.~Dahlen, M.~S. Blok, M.~E. Kimchi-Schwartz, J.~R. McClean, J.~Carter, W.~A. de~Jong, and I.~Siddiqi, ``Computation of molecular spectra on a quantum processor with an error-resilient algorithm,'' {\em Phys. Rev. X}, vol.~8, p.~011021, Feb 2018.

\bibitem{parrish2019filterdiagonalization}
R.~M. Parrish and P.~L. McMahon, ``Quantum filter diagonalization: Quantum eigendecomposition without full quantum phase estimation,'' {\em arxiv preprint, {\rm arxiv:1909.08925}}, 2019.

\bibitem{motta2020qite_qlanczos}
M.~Motta, C.~Sun, A.~T.~K. Tan, M.~J. O'Rourke, E.~Ye, A.~J. Minnich, F.~G. S.~L. Brand{\~a}o, and G.~K.-L. Chan, ``Determining eigenstates and thermal states on a quantum computer using quantum imaginary time evolution,'' {\em Nature Physics}, vol.~16, no.~2, pp.~205--210, 2020.

\bibitem{takeshita2020subspace}
T.~Takeshita, N.~C. Rubin, Z.~Jiang, E.~Lee, R.~Babbush, and J.~R. McClean, ``Increasing the representation accuracy of quantum simulations of chemistry without extra quantum resources,'' {\em Phys. Rev. X}, vol.~10, p.~011004, Jan 2020.

\bibitem{huggins2020nonorthogonal}
W.~J. Huggins, J.~Lee, U.~Baek, B.~O'Gorman, and K.~B. Whaley, ``A non-orthogonal variational quantum eigensolver,'' {\em New Journal of Physics}, vol.~22, p.~073009, jul 2020.

\bibitem{stair2020krylov}
N.~H. Stair, R.~Huang, and F.~A. Evangelista, ``A multireference quantum {K}rylov algorithm for strongly correlated electrons,'' {\em Journal of Chemical Theory and Computation}, vol.~16, pp.~2236--2245, 04 2020.

\bibitem{urbanek2020chemistry}
M.~Urbanek, D.~Camps, R.~Van~Beeumen, and W.~A. de~Jong, ``Chemistry on quantum computers with virtual quantum subspace expansion,'' {\em Journal of Chemical Theory and Computation}, vol.~16, pp.~5425--5431, 09 2020.

\bibitem{kyriienko2020quantum}
O.~Kyriienko, ``Quantum inverse iteration algorithm for programmable quantum simulators,'' {\em npj Quantum Information}, vol.~6, no.~1, p.~7, 2020.

\bibitem{cohn2021filterdiagonalization}
J.~Cohn, M.~Motta, and R.~M. Parrish, ``Quantum filter diagonalization with compressed double-factorized {H}amiltonians,'' {\em PRX Quantum}, vol.~2, p.~040352, Dec 2021.

\bibitem{yoshioka2021virtualsubspace}
N.~Yoshioka, H.~Hakoshima, Y.~Matsuzaki, Y.~Tokunaga, Y.~Suzuki, and S.~Endo, ``Generalized quantum subspace expansion,'' {\em Phys. Rev. Lett.}, vol.~129, p.~020502, Jul 2022.

\bibitem{epperly2021subspacediagonalization}
E.~N. Epperly, L.~Lin, and Y.~Nakatsukasa, ``A theory of quantum subspace diagonalization,'' {\em SIAM Journal on Matrix Analysis and Applications}, vol.~43, no.~3, pp.~1263--1290, 2022.

\bibitem{seki2021powermethod}
K.~Seki and S.~Yunoki, ``Quantum power method by a superposition of time-evolved states,'' {\em PRX Quantum}, vol.~2, p.~010333, Feb 2021.

\bibitem{bespalova2021hamiltonian}
T.~A. Bespalova and O.~Kyriienko, ``Hamiltonian operator approximation for energy measurement and ground-state preparation,'' {\em PRX Quantum}, vol.~2, p.~030318, Aug 2021.

\bibitem{baker2021lanczos}
T.~E. Baker, ``{L}anczos recursion on a quantum computer for the {G}reen's function and ground state,'' {\em Phys. Rev. A}, vol.~103, p.~032404, Mar 2021.

\bibitem{baker2021block}
T.~E. Baker, ``Block {L}anczos method for excited states on a quantum computer,'' {\em Phys. Rev. A}, vol.~110, p.~012420, Jul 2024.

\bibitem{cortes2022krylov}
C.~L. Cortes and S.~K. Gray, ``Quantum {K}rylov subspace algorithms for ground- and excited-state energy estimation,'' {\em Phys. Rev. A}, vol.~105, p.~022417, Feb 2022.

\bibitem{klymko2022realtime}
K.~Klymko, C.~Mejuto-Zaera, S.~J. Cotton, F.~Wudarski, M.~Urbanek, D.~Hait, M.~Head-Gordon, K.~B. Whaley, J.~Moussa, N.~Wiebe, W.~A. de~Jong, and N.~M. Tubman, ``Real-time evolution for ultracompact hamiltonian eigenstates on quantum hardware,'' {\em PRX Quantum}, vol.~3, p.~020323, May 2022.

\bibitem{jamet2022greens}
F.~Jamet, A.~Agarwal, and I.~Rungger, ``Quantum subspace expansion algorithm for {G}reen's functions,'' {\em arxiv preprint, {\rm arxiv:2205.00094}}, 2022.

\bibitem{baek2023nonorthogonal}
U.~Baek, D.~Hait, J.~Shee, O.~Leimkuhler, W.~J. Huggins, T.~F. Stetina, M.~Head-Gordon, and K.~B. Whaley, ``Say no to optimization: A nonorthogonal quantum eigensolver,'' {\em PRX Quantum}, vol.~4, p.~030307, Jul 2023.

\bibitem{tkachenko2022davidson}
N.~V. Tkachenko, L.~Cincio, A.~I. Boldyrev, S.~Tretiak, P.~A. Dub, and Y.~Zhang, ``Quantum {D}avidson algorithm for excited states,'' {\em Quantum Science and Technology}, vol.~9, p.~035012, apr 2024.

\bibitem{lee2023sampling}
G.~Lee, D.~Lee, and J.~Huh, ``Sampling {E}rror {A}nalysis in {Q}uantum {K}rylov {S}ubspace {D}iagonalization,'' {\em {Quantum}}, vol.~8, p.~1477, Sept. 2024.

\bibitem{zhang2023measurementefficient}
Z.~Zhang, A.~Wang, X.~Xu, and Y.~Li, ``Measurement-efficient quantum {K}rylov subspace diagonalisation,'' {\em {Quantum}}, vol.~8, p.~1438, Aug. 2024.

\bibitem{kirby2023exactefficient}
W.~Kirby, M.~Motta, and A.~Mezzacapo, ``Exact and efficient {L}anczos method on a quantum computer,'' {\em {Quantum}}, vol.~7, p.~1018, May 2023.

\bibitem{shen2023realtimekrylov}
Y.~Shen, K.~Klymko, J.~Sud, D.~B. Williams-Young, W.~A.~d. Jong, and N.~M. Tubman, ``Real-{T}ime {K}rylov {T}heory for {Q}uantum {C}omputing {A}lgorithms,'' {\em {Quantum}}, vol.~7, p.~1066, July 2023.

\bibitem{yang2023dualgse}
B.~Yang, N.~Yoshioka, H.~Harada, S.~Hakkaku, Y.~Tokunaga, H.~Hakoshima, K.~Yamamoto, and S.~Endo, ``Resource-efficient generalized quantum subspace expansion,'' {\em Phys. Rev. Appl.}, vol.~23, p.~054021, May 2025.

\bibitem{yang2023shadow}
R.~Yang, T.~Wang, B.-N. Lu, Y.~Li, and X.~Xu, ``Shadow-based quantum subspace algorithm for the nuclear shell model,'' {\em Phys. Rev. A}, vol.~111, p.~012620, Jan 2025.

\bibitem{ohkura2023leveraging}
Y.~Ohkura, S.~Endo, T.~Satoh, R.~V. Meter, and N.~Yoshioka, ``Leveraging hardware-control imperfections for error mitigation via generalized quantum subspace,'' {\em arxiv preprint, {\rm arxiv:2303.07660}}, 2023.

\bibitem{motta2023subspace}
M.~Motta, W.~Kirby, I.~Liepuoniute, K.~J. Sung, J.~Cohn, A.~Mezzacapo, K.~Klymko, N.~Nguyen, N.~Yoshioka, and J.~E. Rice, ``Subspace methods for electronic structure simulations on quantum computers,'' {\em Electronic Structure}, vol.~6, no.~1, p.~013001, 2024.

\bibitem{anderson2024solving}
L.~W. Anderson, M.~Kiffner, T.~O'Leary, J.~Crain, and D.~Jaksch, ``Solving lattice gauge theories using the quantum {K}rylov algorithm and qubitization,'' {\em {Quantum}}, vol.~9, p.~1669, Mar. 2025.

\bibitem{kirby2024analysis}
W.~Kirby, ``Analysis of quantum {K}rylov algorithms with errors,'' {\em {Quantum}}, vol.~8, p.~1457, Aug. 2024.

\bibitem{yoshioka2024diagonalization}
N.~Yoshioka, M.~Amico, W.~Kirby, P.~Jurcevic, A.~Dutt, B.~Fuller, S.~Garion, H.~Haas, I.~Hamamura, A.~Ivrii, R.~Majumdar, Z.~Minev, M.~Motta, B.~Pokharel, P.~Rivero, K.~Sharma, C.~J. Wood, A.~Javadi-Abhari, and A.~Mezzacapo, ``Krylov diagonalization of large many-body hamiltonians on a quantum processor,'' {\em Nature Communications}, vol.~16, no.~1, p.~5014, 2025.

\bibitem{byrne2024super}
A.~Byrne, W.~Kirby, K.~M. Soodhalter, and S.~Zhuk, ``A quantum super-{K}rylov method for ground state energy estimation,'' {\em arxiv preprint, {\rm arXiv:2412.17289}}, 2024.

\bibitem{yu2025quantum}
J.~Yu, J.~R. Moreno, J.~T. Iosue, L.~Bertels, D.~Claudino, B.~Fuller, P.~Groszkowski, T.~S. Humble, P.~Jurcevic, W.~Kirby, T.~A. Maier, M.~Motta, B.~Pokharel, A.~Seif, A.~Shehata, K.~J. Sung, M.~C. Tran, V.~Tripathi, A.~Mezzacapo, and K.~Sharma, ``Quantum-centric algorithm for sample-based {K}rylov diagonalization,'' {\em arxiv preprint, {\rm arXiv:2501.09702}}, 2025.

\bibitem{gauss1815methodus}
C.~F. Gauss, ``Methodus nova integralium valores per approximationem inveniendi,'' {\em Comm. Soc. Sci. Göttingen Math.}, vol.~3, pp.~29--76, 1815.

\bibitem{szego1939orthogonal}
G.~Szeg\"o, {\em Orthogonal polynomials}, vol.~XXIII of {\em Colloquium Publications}.
\newblock American Mathematical Society, New York, NY, 1939.

\bibitem{wilf1962mathematics}
H.~S. Wilf, {\em Mathematics for the Physical Sciences}.
\newblock Dover Publications Inc., New York, NY, 1962.

\bibitem{jones1989moment}
W.~B. Jones, O.~Nj{\aa}stad, and W.~J. Thron, ``Moment theory, orthogonal polynomials, quadrature, and continued fractions associated with the unit circle,'' {\em Bulletin of the London Mathematical Society}, vol.~21, no.~2, pp.~113--152, 1989.

\bibitem{gragg1993positive}
W.~B. Gragg, ``Positive definite toeplitz matrices, the arnoldi process for isometric operators, and gaussian quadrature on the unit circle,'' {\em Journal of Computational and Applied Mathematics}, vol.~46, no.~1, pp.~183--198, 1993.

\bibitem{zhang2021quantum}
K.~Zhang, M.-H. Hsieh, L.~Liu, and D.~Tao, ``Quantum gram-schmidt processes and their application to efficient state readout for quantum algorithms,'' {\em Phys. Rev. Res.}, vol.~3, p.~043095, Nov 2021.

\bibitem{kahan2011nearest}
W.~Kahan, ``The nearest orthogonal or unitary matrix,'' {\em University of California at Berkeley, available at: https://people. eecs. berkeley. edu/wkahan/Math128/NearestQ. pdf}, 2011.

\bibitem{sirui2021}
S.~Lu, M.~C. Ba\~nuls, and J.~I. Cirac, ``Algorithms for quantum simulation at finite energies,'' {\em PRX Quantum}, vol.~2, p.~020321, May 2021.

\bibitem{Ghanem2023robustextractionof}
K.~Ghanem, A.~Schuckert, and H.~Dreyer, ``Robust {E}xtraction of {T}hermal {O}bservables from {S}tate {S}ampling and {R}eal-{T}ime {D}ynamics on {Q}uantum {C}omputers,'' {\em {Quantum}}, vol.~7, p.~1163, Nov. 2023.

\bibitem{shen2024trace}
Y.~Shen, K.~Klymko, E.~Rabani, D.~Camps, R.~V. Beeumen, and M.~Lindsey, ``Simple diagonal state designs with reconfigurable real-time circuits,'' {\em arxiv preprint, {\rm arXiv:2401.04176}}, 2024.

\bibitem{summer2024}
A.~Summer, C.~Chiaracane, M.~T. Mitchison, and J.~Goold, ``Calculating the many-body density of states on a digital quantum computer,'' {\em Phys. Rev. Res.}, vol.~6, p.~013106, Jan 2024.

\bibitem{goh2024directestimationdensitystates}
M.~L. Goh and B.~Koczor, ``Direct estimation of the density of states for fermionic systems,'' {\em arxiv preprint, {\rm arXiv:2407.03414}}, 2024.

\bibitem{sugiura2012}
S.~Sugiura and A.~Shimizu, ``Thermal pure quantum states at finite temperature,'' {\em Phys. Rev. Lett.}, vol.~108, p.~240401, Jun 2012.

\bibitem{sugiura2013}
S.~Sugiura and A.~Shimizu, ``Canonical thermal pure quantum state,'' {\em Phys. Rev. Lett.}, vol.~111, p.~010401, Jul 2013.

\bibitem{jagels2017generalized}
C.~Jagels, L.~Reichel, and T.~Tang, ``Generalized averaged szeg{\"o} quadrature rules,'' {\em Journal of Computational and Applied Mathematics}, vol.~311, pp.~645--654, 2017.

\bibitem{sachdeva2013approximation}
S.~Sachdeva and N.~Vishnoi, ``Approximation theory and the design of fast algorithms,'' {\em arxiv preprint, {\rm arxiv:1309.4882}}, 2013.

\end{thebibliography}


\appendix

\begin{widetext}

\section{Definitions, mathematical background, and proofs}
\label[appendix]{quantum_iso_arnoldi_appendix}

For convenience, we restate \cref{quantum_iso_arnoldi} from the main text:

\begin{myalgorithm}[``Quantum isometric Arnoldi method"]~
\label{quantum_iso_arnoldi_app}
\begin{algorithmic}[1]

    \State\textbf{input} $n$-qubit unitary operator $U$, Krylov dimension $d$, initial state $|\psi_0\rangle$

    \For{$i,j=0,1,2,\ldots,d-1$}

    \State Calculate $\textbf{U}_{ij}=\langle\psi_0|U^{j-i+1}|\psi_0\rangle$ \Comment{on quantum (standard Krylov circuits)}

    \State Calculate $\textbf{S}_{ij}=\langle\psi_0|U^{j-i}|\psi_0\rangle$ \Comment{on quantum (standard Krylov circuits)}

    \EndFor

    \State $\widetilde{\textbf{U}} \gets$ Gram-Schmidt orthogonalization of $\textbf{U}$ for Gram matrix $\textbf{S}$

    \State $\widetilde{\textbf{U}}_{,d-1} \gets \widetilde{\textbf{U}}_{,d-1}/\|\widetilde{\textbf{U}}_{,d-1}\|$ \Comment{$\widetilde{\textbf{U}}_{,d-1}$ denotes last column of $\widetilde{\textbf{U}}$}

    \State\textbf{return} $\widetilde{\textbf{U}}$
    
\end{algorithmic}
\end{myalgorithm}

Ref.~\cite{jagels2017generalized} contains a nice review of the following preliminary definitions:
\begin{definition}
    For any complex polynomial $p$ of degree $d$, the \emph{reciprocal polynomial} $p^*$ is defined as
    \begin{equation}
    \label{reciprocal_poly}
        p^*(z)=z^d\overline{p}(z^{-1}),
    \end{equation}
    where $\overline{p}(z)$ denotes complex conjugation of only the coefficients in $p$.
\end{definition}
\begin{definition}
    The \emph{Szeg\"o polynomials} $\phi_k$ associated to $d\mu$ are the sequence of orthogonal polynomials with respect to integration around $\mathds{T}$ with the measure $d\mu$.
    The corresponding \emph{para-orthogonal polynomials} $\widetilde{\phi}_k$ are defined as
    \begin{equation}
    \label{para_orthogonal_poly}
        \widetilde{\phi}_k(z)\coloneqq\phi_k(z)+\tau\phi^*_k(z),
    \end{equation}
    where $\tau\in\mathds{T}$ is a parameter.
\end{definition}
\begin{definition}
    A \emph{Laurent polynomial} of degree $d$ is a function of the form
    \begin{equation}
    \label{laurent_poly_app}
        f(z)=\sum_{j=-d}^d\alpha_jz^j.
    \end{equation}
    If at least one of $\alpha_d,\alpha_{-d}\neq0$, then we say that it is exact degree $d$.
\end{definition}

\begin{definition}
    A \emph{Szeg\"o quadrature rule} of degree $d$ is a set of $d$ nodes $\lambda_k$ and weights $\omega_k$ such that
    \begin{equation}
    \label{szego_quadrature_app}
        \int_\mathds{T}f(x)d\mu(x)=\sum_{k=1}^d\omega_kf(\lambda_k)
    \end{equation}
    holds for all Laurent polynomials $f$ of degree up to $d-1$.
\end{definition}

\begin{lemma}
\label{nodes_are_zeros}
    The nodes $\lambda_k$ of the Szeg\"o quadrature rule \eqref{szego_quadrature_app} are the zeros of the degree-$d$ para-orthogonal polynomial $\widetilde{\phi}_d$. 
\end{lemma}
\begin{proof}
    Proved in Theorem 8.4 in~\cite{jones1989moment}.
\end{proof}

\begin{lemma}
\label{gragg_results}
    $\widetilde{\textbf{U}}^T$ has $d$ distinct eigenvalues, and they are the zeros of the degree-$d$ para-orthogonal polynomial $\widetilde{\phi}_d$.
    An eigenvector corresponding to eigenvalue $\lambda_k$ is
    \begin{equation}
        \vec{\phi}(\lambda_k)\coloneqq(\phi_0(\lambda_k),\phi_1(\lambda_k),\ldots,\phi_{d-1}(\lambda_k))^T.
    \end{equation}
\end{lemma}
\begin{proof}
    These results are shown in~\cite{gragg1993positive}, with their statements on pp.~195 and 196.
\end{proof}

\textbf{Theorem~\ref{main_thm}.}
\emph{
    The nodes $\lambda_k$ of the Szeg\"o quadrature rule \eqref{szego_quadrature_app} are the eigenvalues of $\widetilde{\textbf{U}}^T$.
    The corresponding weights $\omega_k$ are given by
    \begin{equation}
    \label{weights_formula_app}
        \omega_k=|v_{k,0}|^2,
    \end{equation}
    where $v_{k,0}$ is the zeroth entry of a normalized eigenvector $v_k$ of $\widetilde{\textbf{U}}^T$ with eigenvalue $\lambda_k$.
}
\begin{proof}
    The claim regarding nodes follows immediately by combining \cref{nodes_are_zeros} and \cref{gragg_results}.
    The claimed expression \eqref{weights_formula_app} for the weights can in principle be obtained from~\cite{gragg1993positive} as well, as noted in~\cite{jagels2017generalized}, but a lot of work is left for the reader.
    We therefore provide a more detailed proof for the sake of completeness.
    
    From \cref{gragg_results} we have that $\vec{\phi}(\lambda_k)$ is an eigenvector of $\widetilde{\textbf{U}}^T$ corresponding to eigenvalue $\lambda_k$.
    Also, we can directly diagonalize $\widetilde{\textbf{U}}^T$ and obtain the normalized eigenvector $\vec{v}_k$ corresponding to $\lambda_k$.
    Hence $\vec{\phi}(\lambda_k)\propto\vec{v}_k$, so
    \begin{equation}
    \label{vector_length}
        \|\vec{\phi}(\lambda_k)\|=\frac{\|\vec{\phi}(\lambda_k)\|}{\|\vec{v}_k\|}=\frac{|\phi_0(\lambda_k)|}{|v_{k,0}|}=\frac{1}{|v_{k,0}|}
    \end{equation}
    since the zeroth entry of $\vec{\phi}(\lambda_k)$ is $\phi_0(\lambda_k)=1$.

    Note that the remainder of the proof roughly follows the corresponding argument for Gauss quadrature as presented in \cite[Section 2.9]{wilf1962mathematics}.
    The main modifications are due to the different Christoffel-Darboux identity required for the complex unit circle.
    
    For the weights, we begin with eq.~(11.4.5) in \cite{szego1939orthogonal}, a version of the Christoffel-Darboux identity for polynomials that are orthogonal with respect to measure $d\mu$ on the complex unit circle:
    \begin{equation}
    \label{christoffel_darboux}
        \sum_{j=0}^{d-1}\overline{\phi_j(x)}\phi_j(y) = \frac{\overline{\phi_d^*(x)}\phi_d^*(y)-\overline{\phi_d(x)}\phi_d(y)}{1-\overline{x}y}.
    \end{equation}
    We first evaluate this at $x=\lambda_k$: explicitly inserting the definition \eqref{reciprocal_poly} of the reciprocal polynomials and using the fact that $|\lambda_k|=1$ so that $\overline{\lambda_k}=\lambda_k^{-1}$ yields
    \begin{equation}
        \sum_{j=0}^{d-1}\overline{\phi_j(\lambda_k)}\phi_j(y) = \lambda_k\frac{\overline{\phi_d^*(\lambda_k)}\phi_d^*(y)-\overline{\phi_d(\lambda_k)}\phi_d(y)}{\lambda_k-y}.
    \end{equation}
    Now, from \cref{nodes_are_zeros} we know that $\lambda_k$ is a zero of $\widetilde{\phi}_d$.
    Hence from the definition \eqref{para_orthogonal_poly} we have that
    \begin{equation}
    \label{trick_at_zeros}
        \phi_d^*(\lambda_k)=-\tau^{-1}\phi_d(\lambda_k);
    \end{equation}
    inserting this yields
    \begin{equation}
        \sum_{j=0}^{d-1}\overline{\phi_j(\lambda_k)}\phi_j(y) = -\lambda_k\overline{\phi_d(\lambda_k)}\frac{\tau\phi_d^*(y)+\phi_d(y)}{\lambda_k-y} = \lambda_k\overline{\phi_d(\lambda_k)}\frac{\widetilde{\phi}_d(y)}{y-\lambda_k},
    \end{equation}
    where we also used the fact that $|\tau|=1$.
    Next, we integrate $y$ over $d\mu$ around $\mathds{T}$:
    \begin{equation}
    \label{integral_form_1}
        1 = \lambda_k\overline{\phi_d(\lambda_k)}\int_{\mathds{T}}\frac{\widetilde{\phi}_d(y)}{y-\lambda_k}d\mu(y),
    \end{equation}
    where the left-hand side follows because $\phi_0(y)=\phi_0(\lambda_k)=1$ and the $\phi_j(y)$ are orthonormal with respect to $d\mu$, so that $\int_{\textbf{T}}\phi_j(y)d\mu(y)=\delta_{j,0}$.
    Finally, we note that by Theorem 8.4 and Eq. (7.7) in~\cite{jones1989moment},
    \begin{equation}
        \omega_k=\int_{\mathds{T}}\frac{\widetilde{\phi}_d(y)}{(y-\lambda_k)\widetilde{\phi}_d'(\lambda_k)}d\mu(y),
    \end{equation}
    where $\widetilde{\phi}_d'(\lambda_k)$ denotes the first derivative of $\widetilde{\phi}_d$ evaluated at $\lambda_k$ (so the denominator in the argument of the integral is the first-order expansion of $\widetilde{\phi}_d$ around $\lambda_k$).
    Substituting this into \eqref{integral_form_1} yields
    \begin{equation}
    \label{omega_form_1}
        \omega_k = \frac{1}{\lambda_k\overline{\phi_d(\lambda_k)}\widetilde{\phi}_d'(\lambda_k)}.
    \end{equation}
    
    Next, we evaluate the Christoffel-Darboux formula \eqref{christoffel_darboux} at $x=y=\lambda_k$:
    \begin{equation}
        \sum_{j=0}^{d-1}|\phi_j(\lambda_k)|^2 = \lim_{y\rightarrow\lambda_k}\frac{\overline{\phi_d^*(\lambda_k)}\phi_d^*(y)-\overline{\phi_d(\lambda_k)}\phi_d(y)}{1-\overline{\lambda_k}y}.
    \end{equation}
    We must initially write the right-hand side as a limit since both the numerator and the denominator vanish at $y=\lambda_k$.
    Hence we can apply l'Hopital's rule to evaluate the limit:
    \begin{equation}
        \sum_{j=0}^{d-1}|\phi_j(\lambda_k)|^2 = \lim_{y\rightarrow\lambda_k}\frac{\overline{\tau\phi_d^*(\lambda_k)}\left(\widetilde{\phi}_d'(y)-\phi_d'(y)\right)-\overline{\phi_d(\lambda_k)}\phi_d'(y)}{-\overline{\lambda_k}},
    \end{equation}
    where we used \eqref{para_orthogonal_poly} to insert $\frac{d}{dy}\phi_d^*(y)=\tau^{-1}\widetilde{\phi}_d'(y)-\tau^{-1}\phi_d'(y)$.
    Now we can evaluate the limit, and also use \eqref{trick_at_zeros} to simplify the first term:
    \begin{equation}
        \sum_{j=0}^{d-1}|\phi_j(\lambda_k)|^2 = \frac{-\overline{\phi_d(\lambda_k)}\left(\widetilde{\phi}_d'(\lambda_k)-\phi_d'(\lambda_k)\right)-\overline{\phi_d(\lambda_k)}\phi_d'(\lambda_k)}{-\overline{\lambda_k}} = \lambda_k\overline{\phi_d(\lambda_k)}\widetilde{\phi}_d'(\lambda_k),
    \end{equation}
    where in the second step we used the fact that $|\lambda_k|=1$.
    Inserting this into \eqref{omega_form_1} yields
    \begin{equation}
    \label{omega_form_2}
        \omega_k = \frac{1}{\sum_{j=0}^{d-1}|\phi_j(\lambda_k)|^2} = \frac{1}{\|\vec{\phi}(\lambda_k)\|^2}.
    \end{equation}
    The proof of \eqref{weights_formula_app} is completed by inserting \eqref{vector_length}.
    
\end{proof}

\section{Approximating \texorpdfstring{$e^{-\beta\hat{H}}$}{} by a polynomial in \texorpdfstring{$U$}{}}
\label{app:approximating gibbs state}

\begin{lemma}
\label{gibbs_trig_poly_lemma}
    For any $\gamma,\delta\in(0,1), \lambda>0$, there exists $z_j\in\mathds{C}$ such that
    \begin{equation}
    \label{gibbs_trig_poly}
        \max_{x\in[-\gamma,\gamma]}\left|\sum_{j=-d}^dz_je^{ij\pi x/2}-e^{-\lambda(x+1)}\right|\le\delta,
    \end{equation}
    for $d=2\left\lceil\log\left(\frac{4}{\delta}\right)/(1-\gamma)\right\rceil$ and $\sum_{j=-d}^d|z_j|\le1$.
\end{lemma}
\begin{proof}
    By Lemma 3.4 in~\cite{sachdeva2013approximation}, for some integer $K$ there exists a degree-$K$ polynomial ${p(x)\coloneqq\sum_{k=0}^Ka_kx^k}$ such that
    \begin{equation}
        \max_{x\in[-1,1]}\left|p(x)-e^{-\lambda(x+1)}\right|\le\frac{\delta}{4},
    \end{equation}
    for $\sum_{k}|a_k|\le1$.
    Given this, Lemma 37 in~\cite{vanApeldoorn2020quantumsdpsolvers} immediately implies our desired result.
\end{proof}
\begin{remark}
    By way of explanation, Lemma 37 in~\cite{vanApeldoorn2020quantumsdpsolvers} states that if there exists a polynomial approximating $e^{-\lambda(x+1)}$ with pointwise error $\delta/4$, then there exists a trigonometric polynomial approximating $e^{-\lambda(x+1)}$ with pointwise error $\delta$.
    The degree $d$ of the trigonometric polynomial does not directly depend on the degree $K$ of the polynomial, but instead just depends on the norm of its coefficients, the error $\delta$, and the size $\gamma$ of the window.
    The norm of the coefficients of the trigonometric polynomial is upper bounded by the norm of the coefficients of the polynomial, which in our case is upper bounded by $1$.
    Substituting the corresponding symbols into Lemma 37 in~\cite{vanApeldoorn2020quantumsdpsolvers} is all that is required.
\end{remark}
\begin{remark}
    The proofs of Lemma 3.4 in~\cite{sachdeva2013approximation} and Lemma 37 in~\cite{vanApeldoorn2020quantumsdpsolvers} are both constructive, yielding an efficient algorithm for constructing the expansion coefficients $z_j$.
\end{remark}

We apply \cref{gibbs_trig_poly_lemma} to approximating the Gibbs state by letting $x=\frac{\gamma E}{\|H\|}$ and $\lambda=\frac{\beta\|H\|}{\gamma}$ where $\gamma\in(0,1)$ is a free parameter for now.
With these choices,
\begin{equation}
    e^{-\lambda(x+1)}=e^{-\beta E}e^{-\beta\|H\|/\gamma},
\end{equation}
where the second factor on the right-hand side may be interpreted as a temperature dependent rescaling.
Substituting the above into \eqref{gibbs_trig_poly} yields
\begin{equation}
    \max_{\frac{\gamma E}{\|H\|}\in[-\gamma,\gamma]}\left|\sum_{j=-d}^dz_je^{\frac{ij\pi\gamma E}{2\|H\|}}-e^{-\beta E}e^{-\beta\|H\|/\gamma}\right|\le\delta,
\end{equation}
which simplifies to
\begin{equation}
    \max_{E\in[-\|H\|,\|H\|]}\left|\sum_{j=-d}^d\hat z_je^{ijE\,dt}-e^{-\beta E}\right|\le e^{\beta\|H\|/\gamma}\delta 
\end{equation}
with
\begin{equation}
\label{expansion_params}
\begin{split}
    &dt \coloneqq \frac{\pi\gamma}{2\|H\|}, \\
    &\hat z_j \coloneqq e^{\beta\|H\|/\gamma}z_j.
\end{split}
\end{equation}
This implies that
\begin{equation}
    \left\|\sum_{j=-d}^d\hat z_je^{ijH\,dt}-e^{-\beta H}\right\|\le\delta e^{\frac{\beta\|H\|}{\gamma}},
\end{equation}
since all of the operators involved are diagonal in the eigenbasis of $H$.
Note that for a given $d$ there is a trade-off between $\delta$ and $1-\gamma$, since they must satisfy $d=2\left\lceil\frac{1}{1-\gamma}\log\left(\frac{4}{\delta}\right)\right\rceil$.
Wrapping everything up, in order to achieve an error $\hat\delta$ for a trigonometric polynomial (i.e., Laurent polynomial in $U=e^{iH\,dt}$) approximation of $e^{-\beta H}$, we must choose $dt$ and coefficients as in \eqref{expansion_params}, and degree
\begin{equation}
    d=2\left\lceil\frac{1}{1-\gamma}\log\left(\frac{4e^{\frac{\beta\|H\|}{\gamma}}}{\hat\delta}\right)\right\rceil=2\left\lceil\frac{1}{1-\gamma}\log\left(\frac{4}{\hat\delta}\right) + \frac{\beta\|H\|}{\gamma}\right\rceil.
\end{equation}
Equivalently, the error as a function of $d$ and the other problem parameters is given by
\begin{equation}
    \hat\delta=4\exp\left(\frac{\beta\|H\|}{\gamma}-\frac{d}{2}(1-\gamma)\right).
\end{equation}

\end{widetext}

\end{document}